\documentclass{article} 
\usepackage{iclr2026_conference,times}

\usepackage{amsmath,amsfonts,bm}

\def\eqref#1{equation~\ref{#1}}

\def\1{\bm{1}}

\DeclareMathAlphabet{\mathsfit}{\encodingdefault}{\sfdefault}{m}{sl}
\SetMathAlphabet{\mathsfit}{bold}{\encodingdefault}{\sfdefault}{bx}{n}

\DeclareMathOperator*{\argmax}{arg\,max}

\iclrfinalcopy
\usepackage{hyperref}
\usepackage{graphicx}
\usepackage{url}
\usepackage{amsmath,amssymb,amsthm}
\usepackage{booktabs}
\usepackage{multirow}
\usepackage{siunitx}
\usepackage{makecell}
\usepackage{tabularx}
\usepackage{subcaption}

\theoremstyle{plain} 
\newtheorem{theorem}{Theorem}[section]      
\newtheorem{lemma}[theorem]{Lemma}          
\newtheorem{proposition}[theorem]{Proposition}
\newtheorem{corollary}[theorem]{Corollary}

\title{On Theoretical Interpretations of Concept-Based In-Context Learning}

\author{Huaze Tang$^\dagger$, Tianren Peng$^\dagger$ \& Shao-lun Huang \thanks{$^\dagger$: Huaze Tang and Tianren Peng are of equal contribution. *: Shao-lun Huang is the corrsponding author.} \\
Institute of Data and Information\\
Tsinghua Shenzhen International Graduate School, Tsinghua University\\
Shenzhen, Guangdong, China \\
\texttt{\{tanghz24,ptr22\}@mails.tsinghua.edu.cn, twn2gold@gmail.com}
}

\begin{document}

\maketitle

\begin{abstract}
In-Context Learning (ICL) has emerged as an important new paradigm in natural language processing and large language model (LLM) applications. However, the theoretical understanding of the ICL mechanism remains limited. This paper aims to investigate this issue by studying a particular ICL approach, called concept-based ICL (CB-ICL). In particular, we propose theoretical analyses on applying CB-ICL to ICL tasks, which explains why and when the CB-ICL performs well for predicting query labels in prompts with only a few demonstrations. In addition, the proposed theory quantifies the knowledge that can be leveraged by the LLMs to the prompt tasks, and leads to a similarity measure between the prompt demonstrations and the query input, which provides important insights and guidance for model pre-training and prompt engineering in ICL. Moreover,  the impact of the prompt demonstration size and the dimension of the LLM embeddings in ICL are also explored based on the proposed theory. Finally, several real-data experiments are conducted to validate the practical usefulness of CB-ICL and the corresponding theory. 
\end{abstract} 

\maketitle

\section{Introduction}

With the great successes of large language models (LLMs), In-context learning (ICL) has emerged as a new paradigm for natural language processing (NLP) \citep{brown2020language,chowdhery2023palm,achiam2023gpt}, where LLMs addresses the requested queries in context prompts with a few demonstrations. In contrast to conventional supervised learning, the ICL can perform well in prediction and inference tasks with very few samples by leveraging the semantic knowledge learned from the LLMs without training or fine-tuning the model parameters \citep{liu2022makes,lu2022fantastically,wei2022chain,wu2023self}. This enables rapid task onboarding \citep{sun2022black}, lowers computation and data costs compared with fine-tuning, and underpins current practice in instruction following \citep{linunlocking}, tool use \citep{schick2023toolformer}, and agent memory \citep{chhikara2025mem0}. Therefore, a systematically understanding of ICL mechanism has appeared to be important in engineering designs in many areas of LLMs and NLP. 

Recent researches on understanding the ICL mechanism mainly focused on functional modules \citep{olsson2022context,bietti2023birth,wang2023label,li2024closeness}, theoretical interpretation based on Bayesian and gradient descent Views \citep{xie2021explanation,zhou2023algorithms,dai2023can,mahankali2023one}, and learning and information theoretic perspectives \citep{garg2022can,akyureklearning,pan2023context,yang2024context}. In particular, most of such researches concentrated on analyzing specific mathematical models such as linear regression for given functional classes, investigating the asymptotic learning behaviors of ICL, or characterizing different kinds of convergent properties of transformers in gradient descent. However, there still lacks theoretical justification of why ICL can performs well with only very few demonstrations, and some important questions for deeply understanding ICL mechanism remained open, such as theoretically characterizing the knowledge leveraged by LLMs, and quantifying the impact of the prompt engineering \citep{brown2020language} in ICL.

In this paper, we develop a theoretical framework to analyze the performance of the Concept-Based In-Context Learning (CB-ICL) approach to address the aforementioned issues. As illustrated in Figure 1, in CB-ICL, a pre-trained LLM is employed to represent the semantic embeddings of the prompt contexts, where the parameters of the LLM is fixed throughout the learning task without fine-tuning. Then, the concept of the prompt contexts, represented by the vector $\underline{\alpha}$, is learned by the a prompt concept extractor. Finally, the concept vector is applied to estimate the posterior distribution of the label given the query context for solving the prediction task. To characterize the performance of CB-ICL, we establish upper bounds for the mean-squared excessive risk between the estimated posterior distribution and the ground truth distribution. The proposed upper bounds and the corresponding analyses lead to the theoretical insights and contributions summarized as follows:
\begin{figure}[t]
\begin{center}
\includegraphics[width=\textwidth]{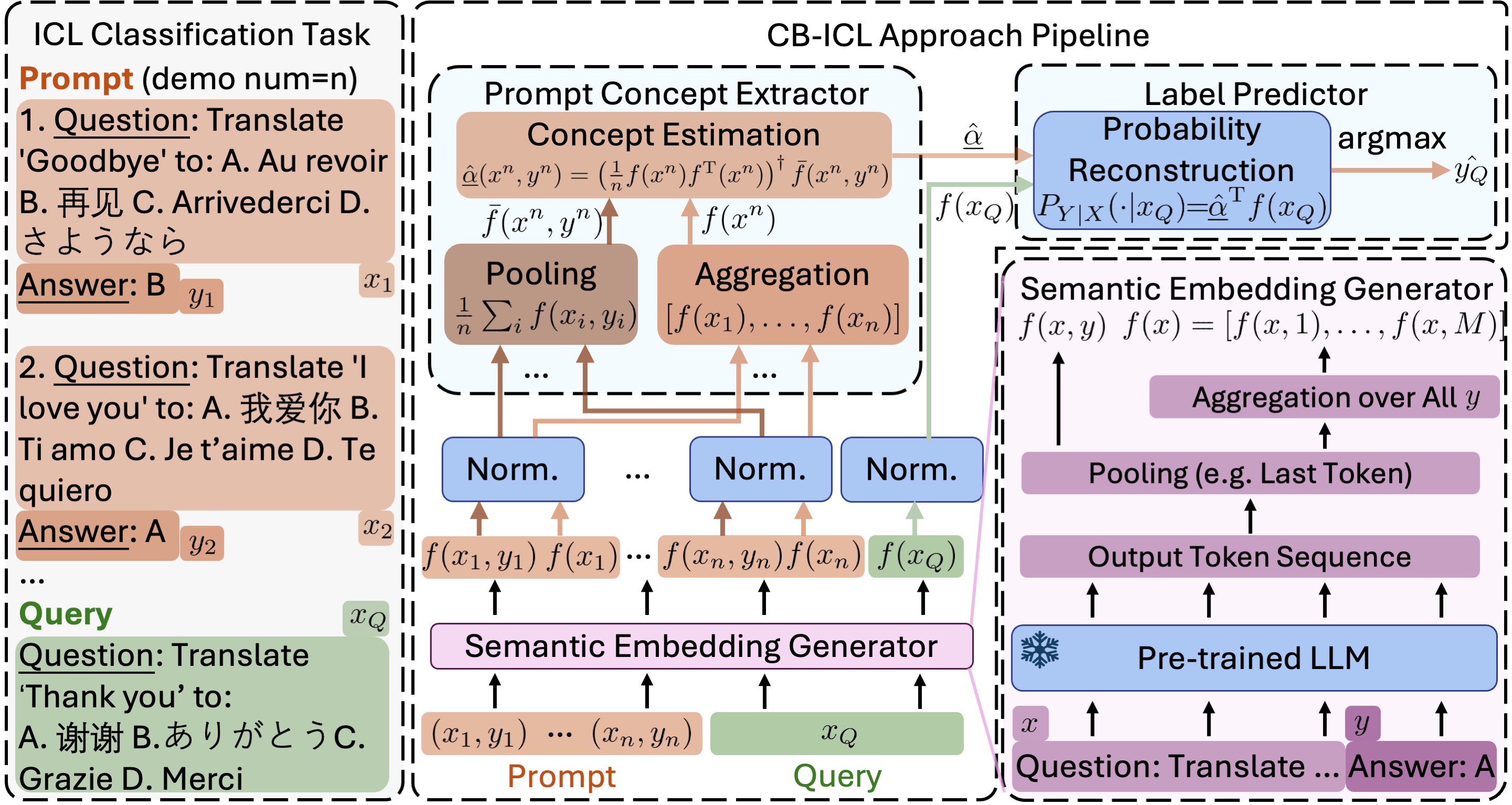}
\end{center}
\caption{The targeting task and working pipeline of the CB-ICL approach.}
\end{figure}

\begin{itemize}
    \item We model the semantic knowledge leveraged from the pre-trained LLM to the CB-ICL as a projection of the prompt distribution onto the semantic knowledge subspace spanned by the LLM embedding, and show that the CB-ICL achieves theoretically provably good learning performance, if both the semantic concept of the prompt is captured by the pre-trained LLM, and the correlation between the prompt query and the corresponding label is strong (which is true in many practical scenarios). This explains why well-designed LLMs can lead to good performances in many ICL applications. 
    \item In addition, a similarity measure between the prompt demonstrations and queries is defined from the derived upper bound, which characterizes the impact of selecting related prompt demonstrations in CB-ICL. Such a similarity measure suggests how to design theoretically good demonstrations in prompt engineering.
    \item Moreover, we demonstrate the impacts of the dimension of LLM embedding, number of prompt demonstrations, and the cardinality of labels in ICL. In particular, it is shown that the higher LLM embedding dimension, the more difficult to learn the prompt concept from the semantic knowledge subspace of the LLM embedding, which suggests the importance of constructing parsimony and informative LLM embeddings in ICL.
    \item Furthermore, we quantify the learning performance degradation when the prompt demonstrations are not sufficient to illustrate the prompt concept, or the LLM embedding cannot fully capture the semantic knowledge of the prompt, which provides theoretical insights of prompt engineering and pre-training.
    \item Finally, real-data experiments are conducted on several LLMs and datasets to validate the performance of CB-ICL. The results show that the performance of CB-ICL is comparable to the existing ICL methods. Moreover, the aforementioned theoretical insights are also verified, which leads to useful guidance for designing effective ICL approaches.
\end{itemize}

\section{Related Works}
\subsection{In-context Learning Mechanism}

Recent studies have sought to explain why LLMs exhibit strong ICL abilities. One perspective interprets ICL as implicit meta-learning or Bayesian inference \cite{xie2021explanation,dai2023can,li2023transformers}. In this view, the prompt defines a latent task, and the transformer adapts its predictions by inferring the underlying hypothesis that best explains the demonstration–query pairs. Related works argue that transformers approximate Bayesian model averaging over latent tasks, aligning ICL with posterior inference in probabilistic models \citep{zhou2023algorithms}.

Another line of work emphasizes algorithmic simulation \citep{garg2022can,akyureklearning,pan2023context}. Using synthetic regression tasks, \citet{garg2022can} showed that pre-trained transformers can reproduce standard estimators such as least-squares or ridge regression. Follow-up analyses demonstrated that multi-layer architectures can in principle implement iterations of gradient descent or even closed-form regression updates \citep{akyureklearning, von2023transformers}, suggesting that deeper layers correspond to successive optimization steps. These findings support the view that transformers can internally encode and update small predictive models when exposed to demonstrations.

The third thread highlights the role of pre-training diversity \citep{yang2024context,raventos2023pretraining}. \citet{raventos2023pretraining} identified threshold phenomena: if the set of pre-training tasks is too narrow, the model behaves like a biased predictor that struggles with new tasks; once the diversity passes a critical level, the model acquires the ability to generalize in-context, approaching the performance of optimal linear estimators. Besides, \citet{yang2024context} gave the the convergence rate of linear and non-linear regression tasks. This sheds light on why scaling model capacity and broadening training corpora are essential for robust ICL.

Despite these advances, most of the insights are derived in simplified settings such as linear regression or asymptotic regimes. Real-world LLMs operate in rich semantic spaces, where demonstrations and queries may differ substantially in distribution. Explaining how LLMs leverage pre-trained representations to adapt to new tasks with only a handful of examples remains an open theoretical challenge.

\subsection{Prompt Engineering and Demonstration Selection Strategies}

Alongside mechanism studies, a growing body of work has examined how prompt design influences ICL performance. It is well documented that output quality is highly sensitive to the format, order, and content of demonstrations \citep{wang2023large,liu2024let}. Small changes in prompt construction can lead to substantial performance fluctuations, underscoring that prompt crafting is non-trivial.

A straightforward strategy is to retrieve demonstrations semantically similar to the query. Embedding-based nearest-neighbor methods \citep{su2022selective, qin2023context,li2023finding} typically improve accuracy but often introduce redundancy, limiting coverage of the task space. To address this, \citet{su2022selective} and \citet{wang2023large} emphasized the importance of balancing relevance with diversity, ensuring that demonstrations span different labels or sub-concepts. Beyond heuristics, other approaches fine-tune retriever models \citep{rubin2021learning, mavromatis2023examples}, employ reinforcement learning to optimize selection policies \citep{zhang2022active}, or leverage chain-of-thought prompts to guide reasoning \citep{wei2022chain}.

Despite these practical advances, most methods remain heuristic, providing limited theoretical understanding of why certain demonstrations are more effective for a given query. Developing principled criteria for demonstration selection remains an active research frontier, with direct implications for both model pre-training and prompt engineering.

\section{CB-ICL Model Formulation}
\label{CB-ICL}

\paragraph{In Context Prompt Assumptions} 
The prompt contexts are consisted of a collection of $n$ demonstrations $\left\{(x_i,y_i)\right\}_{i=1}^{n}$, where $x_i$’s are the input texts, and $y_i$’s are the corresponding answers such that $y_i \in \{1, \ldots, M\}, \forall i$, and a query input text $x_Q$. Denote $x^n=(x_1,\dots,x_n)$ and $y^n=(y_1,\dots,y_n)$, the goal of ICL is to predict the label $y_Q$ of $x_Q$, given the prompt contexts $x^n$,$y^n$, and $x_Q$ \citep{yang2024context}.  In addition, we assume that $\left\{(x_i,y_i)\right\}_{i=1}^{n}$ and $(x_Q,y_Q)$ are generated and follow the probabilistic relationship:
\begin{equation*}
    P(x_Q,y_Q,y_n|x^n) = P_{X}(x_Q)P_{Y|X}(y_Q|x_Q) \prod_{i=1}^n P_{Y|X}(y_i|x_i),
\end{equation*}
for some ground truth distributions $P_{Y|X}$, and $P_X$, i.e., the labels $y^n$ and $y_Q$ are conditionally independently generated from $x^n$ and $x_Q$ by the same conditional distribution $P_{Y|X}(y|x)$, respectively. Note that we do not make assumptions on the joint distribution of the input texts $x_1,\dots,x_n$, such as independent and identically distributed (i.i.d.), since such assumptions are often unrealistic in practical applications.

\paragraph{Semantic Embeddings} 

Given a pair of input text and label $(x,y)$, we denote the semantic embedding generated by the pre-trained LLM as $f(x,y;\underline{\theta})$, where $\underline{\theta}$ represents the parameters of the LLM. Since there is no parameter fine-tuning in ICL, the parameters of the LLM will be fixed throughout the whole learning tasks, and the semantic embedding will simply be denoted as $f(x,y)$. In addition, before feeding into the prompt concept extractor, the LLM embedding is normalized and a bias term is padded at the end of the embedding vector, i.e., if we denote $f(x,y) = [f_1(x,y), \ldots , f_K(x,y)]$, then  $\sum_{y}f_k(x,y) = 0, \  \sum_{y}f^2_k(x,y) = 1, \ \forall k < K-1$, and $f_K(x,y) = \frac{1}{\sqrt{M}}, \forall x,y$.

Moreover, we express the ground truth distributions $P_{Y|X}$ in terms of a projection onto the space spanned by the LLM embedding functions $f_k(x,y)$, for $k = 1, \ldots, K$, as
\begin{equation}\label{equ:model}
    P_{Y|X}(y|x) = \sum_{k=1}^K \alpha_k f_k(x,y) + R(x,y)=  \underline{\alpha}^{\text{T}} f(x,y) + R(x,y) , \quad \forall x,y, 
\end{equation}
where $\underline{\alpha}$ can be viewed as ground truth concept of the prompt, and $R(x,y)$ is a residual term that can be interpreted as the knowledge not captured by the LLM, and is orthogonal to the LLM embeddings, i.e., 
$$
\sum_{x,y}P_X(x)f_k(x,y)R(x,y) = 0, \quad \forall k.
$$
In the remaining parts, we call the pre-trained LLM model \textit{complete} if $R(x,y)= 0 , \forall x,y$, and \textit{incomplete} otherwise. Note that the complete pre-trained LLM model fully captures the semantic knowledge of the prompt contexts.

\paragraph{Prompt Concept Extractor and Label Predictor}

Given the prompt contexts and the LLM embedding, the prompt concept extractor is defined as
\begin{align}\label{equ:def_alpha}
\underline{\hat{\alpha}}(x^n,y^n) \triangleq \mathbf{F}_n^{\dagger}(x^n) \bar{f}_n(x^n,y^n)
\end{align}
where
\begin{equation*}
    \mathbf{F}_n(x^n)  = \frac{1}{n}\sum_{i=1}^n \sum_{y}  f(x_i,y)f^{\top}(x_i,y), \quad \bar{f}_n(x^n,y^n) = \frac{1}{n}\sum_{i=1}^nf(x_i,y_i),
\end{equation*}
and where “$\dagger$” denotes the psuedo-inverse. Note that $\underline{\hat{\alpha}}(x^n,y^n)$  can be interpreted as an estimation of the ground truth concept $\underline{\alpha}$ from the prompt contexts. In the following, we call the prompt demonstrations \textit{sufficient} if $\mathbf{F}_n(x^n)$ is invertible, and \textit{insufficient} otherwise.

Moreover, we denote $\hat{P}_{Y|X} (\cdot|x_Q)$ as an estimation of the ground truth distributions $P_{Y|X}(\cdot|x_Q)$ given the query input text $x_Q$ defined as
\begin{align*}
&\hat{P}_{Y|X} (y|x_Q) =  \sum_{k=1}^K \hat{\alpha}_k (x^n,y^n) f_k(x_Q,y) =  \underline{\hat{\alpha}}^{\text{T}} (x^n,y^n) f(x_Q,y).
\end{align*}
Then, the CB-ICL label predictor is given by $\mathop{\arg\max}_{y} \hat{P}_{Y|X} (y|x_Q)$.

\paragraph{Mean-Squared Excessive Risk}

In particular, we apply the mean-squared risk to measure the difference between $\hat{P}_{Y|X} (\cdot|x_Q)$ and $P_{Y|X}(\cdot|x_Q)$, defined as
\begin{align*}
\ell \left(x^n,y^n;x_Q\right)& \triangleq \sum_{y} \left(\hat{P}_{Y|X}(y|x_Q) -  P_{Y|X}(y|x_Q) \right)^2.
\end{align*}
The in-context learning capability of CB-ICL is measured by the excessive risk (conditioned on the prompt input texts $x^n$), defined as
\begin{align}\label{equ:excessive_risk}
\mathbb{E}_{P_{Y^n|X^n}}\left[ 
\ell \left(x^n,Y^n;x_Q\right) | X^n = x^n\right].
\end{align}

In the next section, we will draw the connection between the excessive risk and the error probability of label prediction, which shows the usefulness of analyzing the excessive risk in ICL scenarios.

\section{Theoretical Analyses of CB-ICL}
\label{thoery}

\subsection{Complete and Sufficient Models}\label{sec:complete_sufficient}

In this subsection, we analyze the excessive risk of CB-ICL in the case $R(x,y)=0$, and $\mathbf{F}_n(x^n)$ is invertible. To delineate the theoretical results, we define the matrices of the LLM embeddings with respect to the prompt input texts as:
\begin{equation*}
    f(x_Q) = \left[ f(x_Q,1) , \ldots , f(x_Q,M)\right]
\in \mathbb R^{K \times M}, \quad \mathbf{F}(x_Q) = f(x_Q)f^\text{T}(x_Q)\in \mathbb R^{K \times K},
\end{equation*}
\begin{equation*}
    f(x^n) = \left[ f(x_1) , \ldots , f(x_n)\right]
\in \mathbb R^{K \times nM}, \quad \mathbf{F}_n(x^n) = \frac1n f(x^n)f^\text{T}(x^n) \in \mathbb R^{K \times K},
\end{equation*}
and the matrices
\begin{equation*}
    \textbf{A}(x_i) = \textsf{diag}\left\{ P_{Y|X}(1|x_i), \ldots, P_{Y|X}(M|x_i) \right\} - \phi_i\phi_i^{\text{T}},\quad 
\end{equation*}
\begin{align}\label{equ:def_Q}
\textbf{A}(x^n) = \textsf{diag}\left\{ \textbf{A}(x_1), \ldots, \textbf{A}(x_n) \right\}
\end{align}
where $\phi_i^{\text{T}} = \left[ P_{Y|X}(1|x_i), \ldots , P_{Y|X}(M|x_i) , \right]$.

\begin{lemma}\label{lem:lambda_1_ub}
The matrix $\textbf{A}(x_i)$ is positive semi-definite, for all $i$, and the largest eigenvalue of $\textbf{A}(x_i)$, denoted as $\lambda_1(\textbf{A}(x_i))$, satisfies
\begin{equation*}
    \lambda_1(\textbf{A}(x_i)) \leq 2P_Y(y_{\max}|x_i)\left(1-P_Y(y_{\max}|x_i)\right),
\end{equation*}
where $P_Y(y_{\max}|x_i) \triangleq \max_yP_Y(y|x_i))$. Moreover, the largest eigenvalue of $\textbf{A}(x^n)$ satisfies 
\begin{equation*}
    \lambda_1(\textbf{A}(x^n)) = \max_i\lambda_1(\textbf{A}(x_i)).
\end{equation*}
\end{lemma}
\begin{proof}
See Appendix \ref{append:lemma1_proof}.
\end{proof}

Then, the following Theorem characterizes the excessive risk in the complete and sufficient case.
\begin{theorem}\label{thm:compelete_sufficient}
When $R(x,y)=0$, and $\mathbf{F}_n(x^n)$ is invertible, the excessive risk defined in (\ref{equ:excessive_risk}) can be bounded as
\begin{align}\label{equ:complete_sufficient_ub}
&\mathbb{E}_{P_{Y^n|X^n}}\left[ 
\ell \left(x^n,Y^n;x_Q\right) | X^n = x^n\right]\leq \frac{K}{n}\lambda_1 \left(\mathbf{F}(x_Q)\mathbf{F}_n^{-1}(x^n) \right) \lambda_1(\textbf{A}(x^n)),
\end{align}
where $\lambda_1$ denotes the largest eigenvalue, and $\textbf{A}(x^n)$ is as defined in 
(\ref{equ:def_Q}).
\end{theorem}
\begin{proof}
    See Appendix \ref{append:theorem1_proof}.
\end{proof}

In particular, the following theoretical insights that can be obtained from Theorem \ref{thm:compelete_sufficient}:
\begin{itemize}
    \item If the input text $x$ and the label $y$ of the prompt are strongly correlated, i.e., the ground truth distribution satisfies $P_{Y|X}(y_{\max}|x)\simeq 1, \forall x$ (e.g. mathematical reasoning tasks), where $P_{Y|X}(y_{\max}|x)$ is as defined in Lemma \ref{lem:lambda_1_ub}, then from Lemma \ref{lem:lambda_1_ub}, it holds for $\lambda_1(\textbf{A}(x_i)) \simeq 0$, and the excessive risk of CB-ICL for learning the ground truth distribution is vanishing. Therefore, the CB-ICL can achieve striking performance in predicting the label of the query input text with only a few prompt demonstrations, if the LLM embedding is complete, and the input text and label are strongly correlated.
    \item In addition, it is shown in Appendix \ref{sec:lb_lambda_1} that $\lambda_1 \left(\mathbf{F}(x_Q)\mathbf{F}_n^{-1}(x^n) \right) \geq 1$ with equality holds when  $\mathbf{F}(x_Q) = \mathbf{F}_n(x^n)$. This provides a theoretical explanation that designing semantically correlated demonstrations in the prompt can enhance the ICL performance. Moreover, the quantity $\lambda_1^{-1} \left(\mathbf{F}(x_Q)\mathbf{F}_n^{-1}(x^n) \right)$ can be employed as a measure for selecting semantically correlated demonstrations to improve the ICL performance in prompt engineering. The real-data experiments in the Section \ref{sec:exp} shows the applicability of this measure in real problem.
    \item Finally, it can be observed from (4) that excessive risk is inversely proportional to the number $n$ of demonstrations in the prompt, and proportional to the dimension $K$ of the LLM embedding. Therefore, although designing high dimensional LLM embeddings can capture more semantic knowledge of the prompt, i.e., making $R(x,y)$ smaller (c.f. \citep{kaplan2020scaling}), in the meanwhile, it is also more difficulty to learn the prompt concept from the high-dimensional embedding space reflected by the growing excessive risk with respect to the dimension $K$. This tradeoff suggest the importance of learning parsimony and informative semantic embeddings in ICL.
\end{itemize}

\subsection{Complete and Insufficient Models}

In this subsection, we investigate the situation when the LLM embedding is complete, but the prompt demonstrations are not sufficient to illustrate the prompt concept, i.e., $\mathbf{F}_n(x^n)$ is not a full-rank matrix. We aim to characterize the impact of insufficient prompt demonstrations in CB-ICL, which can be formalized in the following Theorem.

\begin{theorem}\label{thm:complete_insufficient}
When $R(x,y)=0$, and $\mathbf{F}_n(x^n)$ is not invertible, the excessive risk (2) can be bounded as
\begin{align*}
\mathbb{E}_{P_{Y^n|X^n}}\left[ 
\ell \left(x^n,Y^n;x_Q\right) | X^n = x^n\right]\leq \frac{K}{n}\lambda_1 \left(\mathbf{F}(x_Q)\mathbf{F}_n^\dagger(x^n) \right) \lambda_1(\textbf{A}(x^n)) + \underbrace{
\left\| f ^{\mathrm{T}}(x_Q)\mathbf{F}_n^\perp(x^n)\underline{\alpha}\right\|^2 
}_{(*)},
\end{align*}
where $\mathbf{F}_n^{\perp}(x^n) = \textbf{I}_K-\mathbf{F}_n^\dagger(x^n) \mathbf{F}_n(x^n)$.
\end{theorem}
\begin{proof}
    See Appendix \ref{append:theorem2_proof}.
\end{proof}
Comparing to (\ref{equ:complete_sufficient_ub}) in Theorem \ref{thm:compelete_sufficient}, there is a penalty term $(*)$ caused by the insufficiency of the prompt demonstrations to the upper bound of the excessive risk. To interpret this term, note that from the definition of (\ref{equ:def_alpha}), $\underline{\hat{\alpha}}(x^n,y^n)$ is orthogonal to the null space of $\mathbf{F}_n(x^n)$, and hence the CB-ICL cannot learn the prompt concept aligned with the null space of $\mathbf{F}_n(x^n)$. Since $\mathbf{F}_n^{\perp}(x^n)$ is a projection operation, which projects the vector onto the null space of $\mathbf{F}_n(x^n)$, the penalty term $(*)$ can be interpreted as the information of the query embedding $f (x_Q)$ aligned with the null space of $\mathbf{F}_n(x^n)$, which cannot be learned by CB-ICL due to the insufficiency of the prompt demonstrations.

Moreover, it is readily to show that if $\mathbf{F}(x_Q) = \mathbf{F}_n(x^n)$, then the penalty term $(*)$ is 0. This tells that if the semantic information of the query text is well illustrated by the prompt demonstrations, the columns of the query embedding matrix $f (x_Q)$ is orthogonal to the null space of $\mathbf{F}_n(x^n)$, and there is no information and performance loss caused by the insufficiency of the prompt demonstrations. This again indicates the importance of designing illustrative prompt demonstrations.

\subsection{Incomplete and Insufficient Models}

In practice, the LLM embeddings often cannot capture the complete knowledge of the prompt contexts, i.e., $R(x,y)\neq0$, which introduce the learning performance degradation caused by the incomplete knowledge LLM embeddings. To quantify such performance degradation, we analyze the expected excessive risk with respect to $P_{X}$, and define the vectors 
\begin{equation*}
    \textsf{R}(x_Q) = \left[ R(x_Q,1) , \ldots , R(x_Q,M)\right]^{\text{T}}
\in \mathbb R^{M}, \textsf{R}(x^n) = \left[ \textsf{R}(x_1) , \ldots , \textsf{R}(x_n)\right]^{\text{T}}
\in \mathbb R^{ nM},
\end{equation*}
Then, the excessive risk can be characterized as follows.

\begin{theorem}\label{thm:incomplete_insufficient}
When $R(x,y)\neq0$, and $\mathbf{F}_n(x^n)$ is not invertible, the excessive risk (2) averaged with respect to $P_{X}$ can be bounded as
\begin{align} \notag
&\mathbb{E}_{P_{X}P_{Y^n|X^n}}\left[
\ell \left(x^n,Y^n;x_Q\right) | X^n = x^n\right]\\ 
&\leq \frac{K}{n}\lambda_1 \left(\mathbf{F}_Q\mathbf{F}_n^\dagger(x^n) \right) \lambda_1(\textbf{A}(x^n))\notag\\
& \quad + \frac1n \lambda_1 \big(\mathbf{F}_Q \mathbf{F}_n^\dagger(x^n)\big)\cdot \| \mathsf{R}(x^n) \|^2 + \sum_{x_Q,y}{P_{X}(x_Q)R^2(x_Q,y)}\label{equ:penalty_term_Q}\\
& \quad + \underline{\alpha}^{\mathrm{T}}\mathbf{F}_n^{\perp}(x^n)^{\mathrm{T}} \mathbf{F}_Q\mathbf{F}_n^{\perp}(x^n)\underline{\alpha} - \frac2n \mathsf{R}^{\mathrm{T}}(x^n)f^{\mathrm{T}}(x^n) \mathbf{F}_n^\dagger(x^n)\mathbf{F}_Q\mathbf{F}_n^{\perp}(x^n)\underline{\alpha} \label{equ:penalty_term}
\end{align}
where $\mathbf{F}_Q = \mathbb E_{P_{X}} [\mathbf{F}(X_Q)]$.
\end{theorem}
\begin{proof}
    See Appendix \ref{append:theorem3_proof}.
\end{proof}

Similar to the discussions in Section \ref{sec:complete_sufficient}, the penalty term (\ref{equ:penalty_term}) quantifies the amount of semantic information aligned with the null space of $\mathbf{F}_n(x^n)$ that cannot be learned in the CB-ICL approach due to the insufficiency of the prompt demonstrations. 
Moreover, the penalty terms (\ref{equ:penalty_term_Q}) quantify the performance degradation of the excessive risk caused by the incompleteness of the LLM embedding. In particular, the first term of (\ref{equ:penalty_term_Q}) can be interpreted as the learning bias in the model (\ref{equ:model}), due to the incompleteness of the LLM embedding, and the second term of (\ref{equ:penalty_term_Q}) quantifies the amount of semantic information of the query input text that is not captured by the LLM embeddings. 
Finally, we present some remarks of CB-ICL to draw the connections and comparisons to existing intuitions and techniques of ICL researches:

\begin{itemize}
    \item It is widely believed in ICL researches that during pre-training, LLM models acquire a broad range of semantic prior knowledge from the training data, which later aids task-specific learning representations \citep{chan2022data,shin2022effect,yadlowsky2023pretraining,yang2024context}. In particular, this empirical observation can be theoretically justified by the CB-ICL, which allows the generalizability to a broad class of ICL problems, as long as the ground truth distribution of the prompt contexts is somehow aligned with the semantic knowledge subspace spanned by the LLM embeddings.
    \item The pre-training/warm-up techniques \citep{brunet2023icl,shi2023context,li2024mend} in traditional ICL can also be beneficial in CB-ICL, which reduces the modeling error $R(x,y)$, and improve the learning performance. Moreover, the pre-training loss does not need to be restricted to the MSE loss between $\hat{P}_{Y|X}$ and $P_{Y|X}$, as long as the global minimum of the pre-training loss is achieved at $\hat{P}_{Y|X} = P_{Y|X}$.
    \item Note that the prompt concept extractor can be interpreted as a wide-sense transformer with the softmax activation function replaced by a quadratic function (c.f. linear attetion \citep{wang2020linformer,shen2021efficient,han2024demystify}). This essentially suggests the application of more general kinds of transformer architectures in theoretical analyses and algorithm designs in ICL and other machine learning fields.
\end{itemize}

\subsection{The Label Predicting Error Probability }

To further justify the practically usefulness of the CB-ICL and the corresponding theoretical analyses, in this subsection we establish the connection between the mean-squared excessive risk and the label predicting error probability that is widely adopted in real applications. To this end, we define
\begin{align}\label{equ:label_predictor}
\hat{y}_{\max} = \arg\max_y \hat{P}_{Y|X}(y|x_Q),
\end{align}
and denote $P_j$ as the $j$th largest probability among $\{P_{Y|X}(y|x_Q)\}_{y=1}^M$. In the following, we assume $P_1 > P_j$, for all $j \ge 2$.

\begin{lemma}\label{lem:lebel_pred_error}
Given demonstrations $x^n, y^n$, and query $x_Q$, if the mean-squared risk $\ell(x^n,y^n;x_Q)$ satisfies $\ell(x^n,y^n;x_Q) < \frac12 \left( P_1 - P_{j+1}\right)^2$, for some $j\geq 1$, then
\begin{align}
P_{Y|X}(\hat{y}_{\max}|x_Q)\geq P_j.
\end{align}
\end{lemma}
\begin{proof}
    See Appendix \ref{append:lemma2_proof}.
\end{proof}

Lemma \ref{lem:lebel_pred_error} provides the theoretical guarantee of the CB-ICL label predictor with respect to different threshold values of the mean-squared risk. Notice that $P_{Y|X}(\hat{y}_{\max}|x_Q) \leq \max_y P_{Y|X}(y|x_Q) = P_1$, and the equality is achieved when $\ell(x^n,y^n;x_Q) < \frac12 \left( P_1 - P_{2}\right)^2$. Therefore, the CB-ICL label predictor is reduced to the Maximum a Posteriori (MAP) decision when the mean-squared risk is small. Moreover, the following Theorem establishes the connection between the excessive risk and the label predicting error probability based on Lemma \ref{lem:lebel_pred_error}.
\begin{theorem}\label{thm:error_prob}
Suppose that for some $j \geq 1$, the excessive risk 
\begin{equation*}
    \mathbb{E}_{P_{Y^n|X^n}}\left[ 
\ell \left(x^n,Y^n;x_Q\right) | X^n = x^n\right] = \frac12 (P_1-P_{j})^2 + \gamma,
\end{equation*}
where $0 \leq \gamma < \frac12 (P_1-P_{j+1})^2 - \frac12 (P_1-P_{j})^2$, then the label predicting error probability is lower bounded by
\begin{align}
\mathbb{E}_{P_{Y^n|X^n}}\left[P_{Y|X}(\hat{y}_{\max}|x_Q)\mid X^n=x^n\right] \ge P_j - 
\frac{2\gamma}{2P_1-P_j-P_{j+1}}.
\end{align}
\end{theorem} 

\begin{proof}
    See Appendix \ref{append:theorem4_proof}.
\end{proof}

Theorem \ref{thm:error_prob} shows that designing $\hat{P}_{Y|X}(y|x_Q)$ with small excessive risk also leads to small label predicting error probability from the label predictor Eq. (\ref{equ:label_predictor}), which demonstrates the applicability of the theoretical analyses of the CB-ICL in real scenarios.

\section{Experiment}\label{sec:exp}

We conduct experiments on four representative benchmarks that collectively measure both general knowledge and complex reasoning ability: MMLU (Massive Multitask Language Understanding) \citep{hendrycks2020measuring}, MMLU-Pro (a harder extension of MMLU) \citep{wang2024mmlu}, GPQA (Graduate-level Google-Proof Q\&A) and GPQA-Diamond (the most challenging subset of GPQA) \citep{rein2024gpqa}. These datasets provide a comprehensive evaluation suite, balancing breadth (MMLU, MMLU-Pro) with depth in complex reasoning (GPQA, GPQA-Diamond). Furthermore, we benchmark three representative families of open-source LLMs at different parameter scales (8B, 14B, and 32B), including LLaMA3 family \citep{dubey2024llama}, Qwen3 family \citep{yang2025qwen3} and Deepseek-R1 \citep{guo2025deepseek} distilled model family. This selection covers both general-purpose and reasoning-enhanced model families, enabling a systematic comparison across scaling and architectural choices.

\subsection{Performance Validations of CB-ICL}

\begin{table}[t]
\begin{center}
\resizebox{\textwidth}{!}{
\begin{tabular}{@{}cccccccccc@{}}
\toprule
\multicolumn{2}{c}{\multirow{2}{*}{Dataset}} & \multicolumn{3}{c}{8B} & \multicolumn{2}{c}{14B} & \multicolumn{1}{c}{32B} & \multirow{2}{*}{Average} \\
\cmidrule(lr){3-5}\cmidrule(lr){6-7}\cmidrule(lr){8-8}
 & & LLaMA3 & Qwen3 & Deepseek\text{-}R1 & Qwen3 & Deepseek\text{-}R1 & Qwen3 & \\
\midrule
\multirow{2}{*}{MMLU}
 & ICL \citep{brown2020language}    & 68.40\% & 76.89\% & 63.54\% & 81.05\% & 74.46\% & 83.61\% & 74.66\% \\
 & CB-ICL  & \textbf{71.07\%} & \textbf{77.77\%} & \textbf{64.58\%} & \textbf{81.38\%} & \textbf{80.32\%} & \textbf{83.62\%} & \textbf{76.46\%} \\
\midrule
\multirow{2}{*}{\makecell{MMLU\text{-}\\Pro}}
 & ICL \citep{brown2020language}    & \textbf{35.36\%} & \textbf{56.73\%} & 41.10\% & \textbf{61.03\%} & 57.80\% & 65.54\% & \textbf{52.93\%} \\
 & CB-ICL  & 33.26\% & 53.62\% & \textbf{42.30\%} & 60.47\% & \textbf{59.04\%} & \textbf{65.89\%} & 52.43\% \\
\midrule
\multirow{2}{*}{GPQA}
 & ICL \citep{brown2020language}    & 34.50\% & 44.44\% & 45.32\% & 47.90\% & 46.32\% & 49.49\% & 44.66\% \\
 & CB-ICL  & \textbf{40.77\%} & \textbf{61.60\%} & \textbf{50.82\%} & \textbf{50.65\%} & \textbf{48.45\%} & \textbf{54.82\%} & \textbf{51.19\%} \\
\midrule
\multirow{2}{*}{\makecell{GPQA\text{-}\\diamond}}
 & ICL \citep{brown2020language}    & 28.29\% & 62.00\% & \textbf{49.10\%} & 64.31\% & 59.10\% & \textbf{68.40\%} & 55.20\% \\
 & CB-ICL  & \textbf{33.56\%} & \textbf{63.01\%} & 48.82\% & \textbf{64.88\%} & \textbf{60.13\%} & 66.34\% & \textbf{56.13\%} \\
\bottomrule
\end{tabular}
}
\caption{Performance comparisons on classification task datasets are conducted in both the vanilla ICL \citep{brown2020language} and the CB-ICL setting. We report the accuracy with 5 randomly selected demonstrations from the same task.}
\label{tab:compare}
\end{center}
\end{table}

The results in Table \ref{tab:compare} provide consistent evidence for the effectiveness of CB-ICL. Note that across all model families and datasets, CB-ICL either matches or surpasses vanilla ICL \citep{brown2020language}, denoting the effectiveness of proposed CB-ICL framework. The improvements are particularly notable on GPQA and GPQA-Diamond. This demonstrates that restructuring the prompt to enforce intermediate concepts systematically reduces failure cases. Moreover, improvements are larger on harder datasets. On MMLU (general factual recall), the gains are marginal (less than 2\%), while on GPQA and GPQA-Diamond, the improvements are substantial. This pattern validates the CB-ICL mechanism: when surface recall is sufficient (MMLU), CB-ICL is neutral; when conceptual reasoning is required (GPQA), CB-ICL offers strong advantages. Furthermore, scaling laws persist under CB-ICL. For example, within Qwen3 family, performance consistently increases from 8B to 32B, expect Qwen3-8B on GPQA dataset. This shows that CB-ICL is compatible with scaling, implying that this mechanism complements rather than replaces larger model capacity.

\subsection{Prompt Demonstration Designs}

As discussed in Section \ref{sec:complete_sufficient}, the quantity $\lambda_1^{-1}\left(\mathbf{F}(x_Q)\mathbf{F}^\dagger(x^n)\right)$ can be adopted as a similarity measure between prompt demonstrations and the query input text. For demonstration, we adopt a simple “translate to Chinese” dataset split into two parts that differ only in the cue language (English to Chinese and Italian to Chinese) while sharing the similar answer space and key. We call a demonstration set similar to a query if they come from the same part, otherwise dissimilar. As is illustrated in Figure 2, when the demonstrations are chosen from the similar prompt, the similarity measure $\lambda_1^{-1}\left(\mathbf{F}(x_Q)\mathbf{F}^\dagger(x^n)\right)$ tends to be larger.

\begin{figure}
\centering
\includegraphics[width=\linewidth]{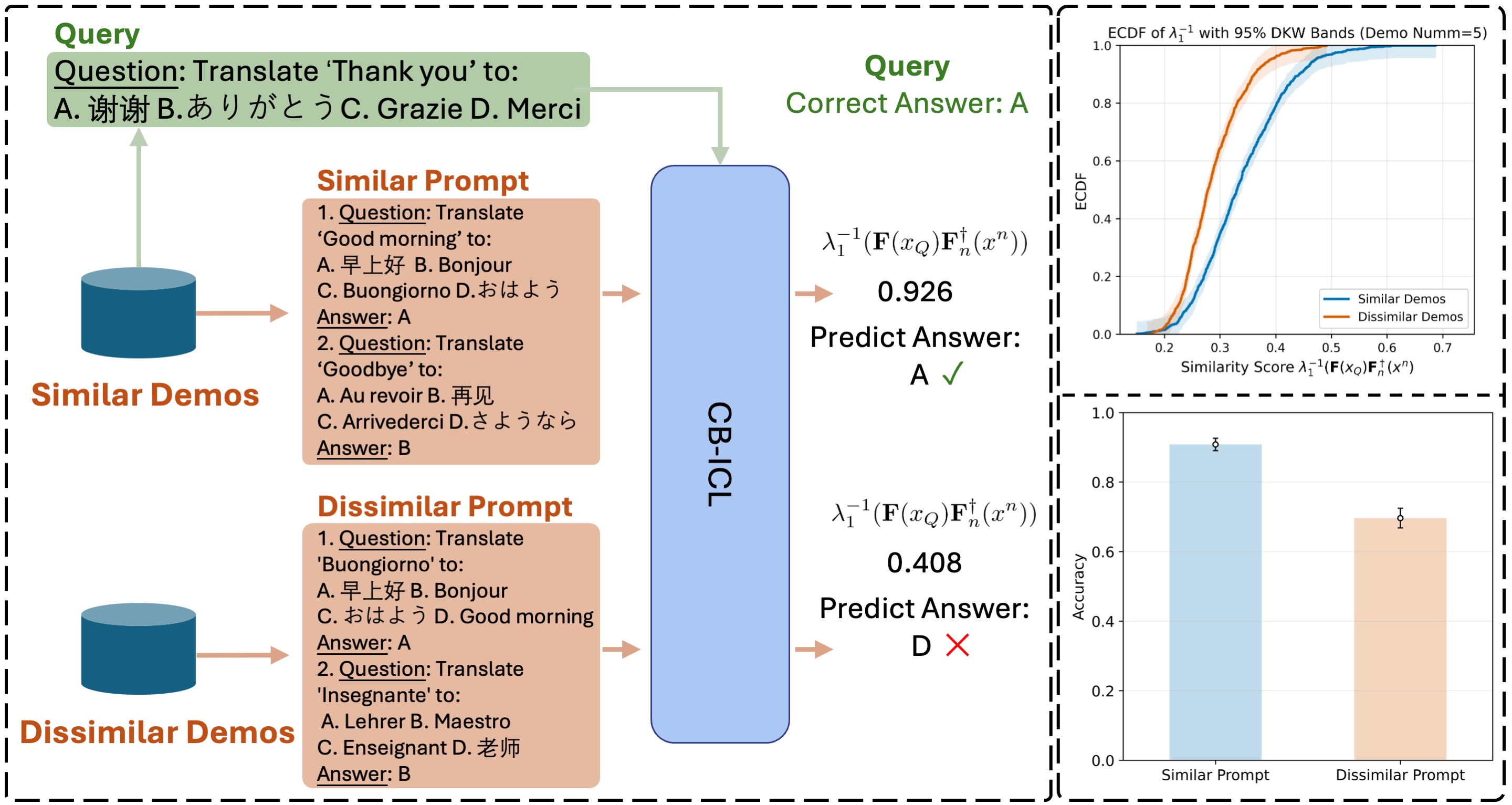}
\caption{Comparison of similarity score $\lambda_1^{-1}\left(\mathbf{F}(x_Q)\mathbf{F}^\dagger(x^n)\right)$ between similar and dissimilar demonstration sets on the “translate to Chinese” task. Demonstrations from similar prompts yield consistently larger similarity score values, indicating stronger semantic alignment with the query.}
\label{fig:golden_exp}
\end{figure}

\begin{table}[t]
\begin{center}
\resizebox{\textwidth}{!}{
\begin{tabular}{@{}cccccccccc@{}}
\toprule
\multicolumn{2}{c}{\multirow{2}{*}{Dataset}} & \multicolumn{3}{c}{8B} & \multicolumn{2}{c}{14B} & \multicolumn{1}{c}{32B} & \multirow{2}{*}{Average} \\
\cmidrule(lr){3-5}\cmidrule(lr){6-7}\cmidrule(lr){8-8}
 & & LLaMA3 & Qwen3 & Deepseek-R1 & Qwen3 & Deepseek-R1 & Qwen3 & \\
\midrule
\multirow{2}{*}{MMLU}
 & CB-ICL              & 71.07\% & 77.77\% & 64.58\% & 81.38\% & 80.32\% & 83.62\% & 76.46\% \\
 & CB-ICL (golden)     & \textbf{73.74\%} & \textbf{80.86\%} & \textbf{68.44\%} & \textbf{82.64\%} & \textbf{82.34\%} & \textbf{84.42\%} & \textbf{78.74\%} \\
\midrule
\multirow{2}{*}{\makecell{MMLU\text{-}\\Pro}}
 & CB-ICL              & 33.26\% & 53.62\% & 42.30\% & 60.47\% & 59.04\% & 65.89\% & 52.43\% \\
 & CB-ICL (golden)     & \textbf{36.58\%} & \textbf{58.97\%} & \textbf{44.24\%} & \textbf{62.67\%} & \textbf{61.34\%} & \textbf{67.29\%} & \textbf{55.18\%} \\
\midrule
\multirow{2}{*}{GPQA}
 & CB-ICL              & 40.77\% & 61.60\% & 50.82\% & 50.65\% & 48.45\% & 54.82\% & 51.19\% \\
 & CB-ICL (golden)     & \textbf{48.51\%} & \textbf{65.77\%} & \textbf{51.67\%} & \textbf{55.83\%} & \textbf{52.32\%} & \textbf{58.73\%} & \textbf{55.47\%} \\
\midrule
\multirow{2}{*}{\makecell{GPQA\text{-}\\diamond}}
 & CB-ICL              & 33.56\% & 63.01\% & 48.82\% & 64.88\% & 60.13\% & 66.34\% & 56.12\% \\
 & CB-ICL (golden)     & \textbf{42.28\%} & \textbf{73.83\%} & \textbf{52.47\%} & \textbf{67.32\%} & \textbf{62.91\%} & \textbf{67.43\%} & \textbf{61.04\%} \\
\bottomrule
\end{tabular}
}
\caption{Performance comparison are conducted between random selected demonstrations and golden selected demonstrations. We report the accuracy with 5 demonstrations from the same task. The original CB-ICL results is generated under 5 randomly picked demonstrations, while CB-ICL (golden) is tested with golden demonstrations. Golden demonstrations are obtained by computing the similarity score between each candidate demonstration and the query using our proposed metric, and selecting the top-5 with the highest scores.}
\label{tab:golden}
\end{center}
\end{table}

To further validate the effectiveness of our proposed similarity measure, we compare CB-ICL under two settings: (i) CB-ICL (random), where demonstrations are randomly selected from the same task, and (ii) CB-ICL (golden), where demonstrations are selected based on the similarity measure. Specifically, for each candidate demonstration $x_i$ in the demonstration pool, we compute its similarity measure $\lambda_1^{-1}(\mathbf{F}(x_Q)\mathbf{F}^\dagger(x_i))$ with the target query $x_Q$. We then rank all candidates and select the top 5 demonstrations with the highest similarity as the golden demonstrations. This procedure ensures that the demonstrations are not only task-relevant but also maximally aligned with the query under our conceptual framework. The results are shown in Table \ref{tab:golden}. We observe three consistent trends: Firstly, golden selection substantially improves accuracy across datasets. For instance, Qwen3-8B achieves 10.82\% improvement on GPQA-Diamond, while LLaMA-7B improves from 33.56\% to 42.28\%. These large margins highlight that the similarity measure successfully identifies high-value demonstrations that guide the model more effectively. Secondly, the gains are most pronounced on more challenging datasets. Improvements are relatively modest on MMLU (less than 3\%), but significantly larger on GPQA and GPQA-Diamond, which require deeper reasoning. This pattern is consistent with our hypothesis: when tasks are non-trivial and prone to spurious correlations, the similarity-guided selection provides stronger benefits. Thirdly, the improvement generalizes across model families and scales. Both LLaMA3, Qwen3 and DeepSeek-R1 benefit from golden selection, regardless of parameter size. Even large models (e.g., Qwen3-32B) see measurable improvements, suggesting that the similarity measure provides complementary guidance beyond model scaling.

\subsection{Impact of Incompleteness of LLMs in ICL}

\begin{figure}[t]
  \centering
  \begin{subfigure}[t]{0.48\linewidth}
    \centering
    \includegraphics[width=\linewidth]{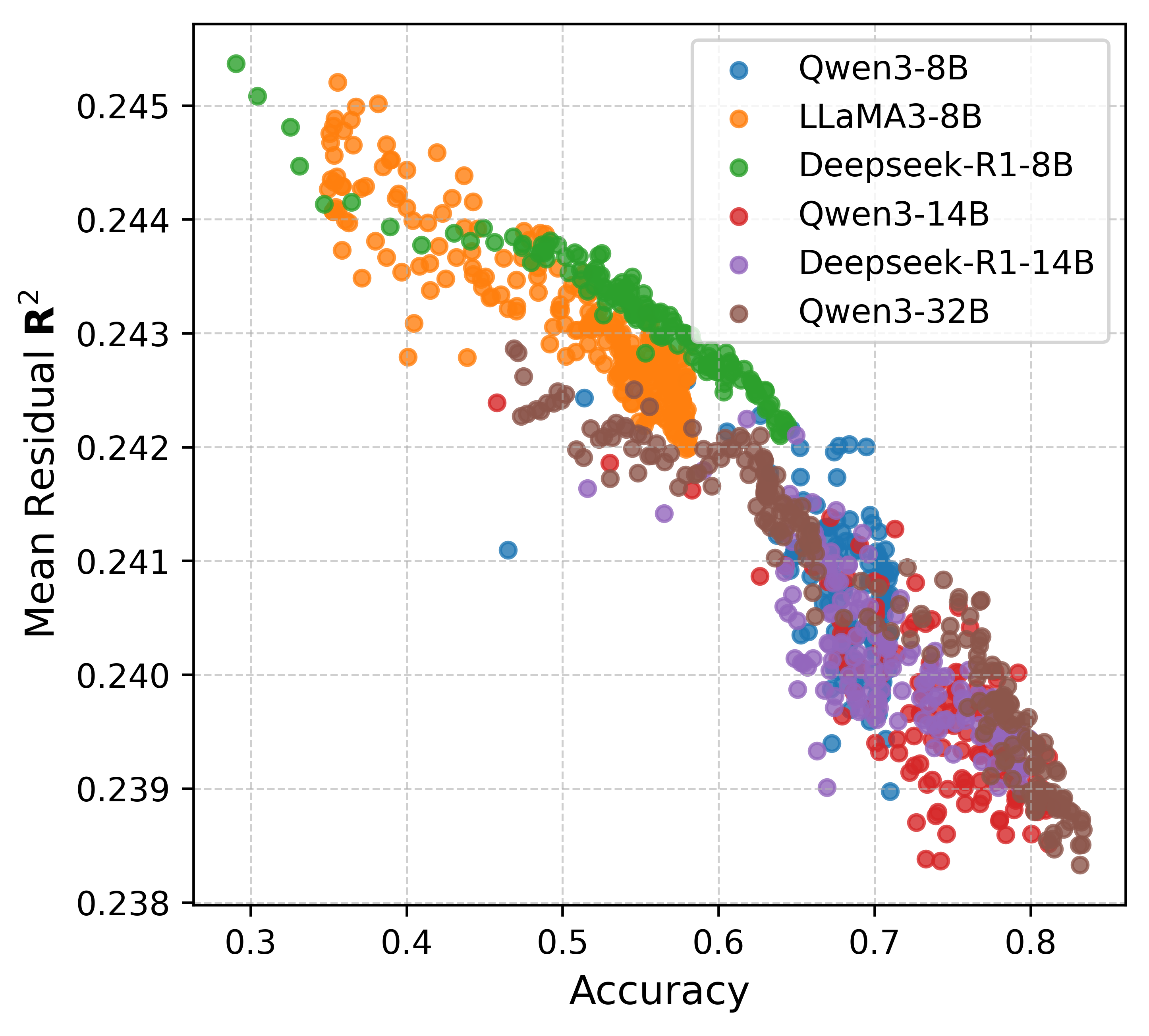}
    \caption{\textbf{MMLU}: Accuracy ($\uparrow$) vs residual risk $\mathbf{R}^2$ ($\downarrow$).}
    \label{fig:incomplete-mmlu}
  \end{subfigure}
  \hfill
  \begin{subfigure}[t]{0.48\linewidth}
    \centering
    \includegraphics[width=\linewidth]{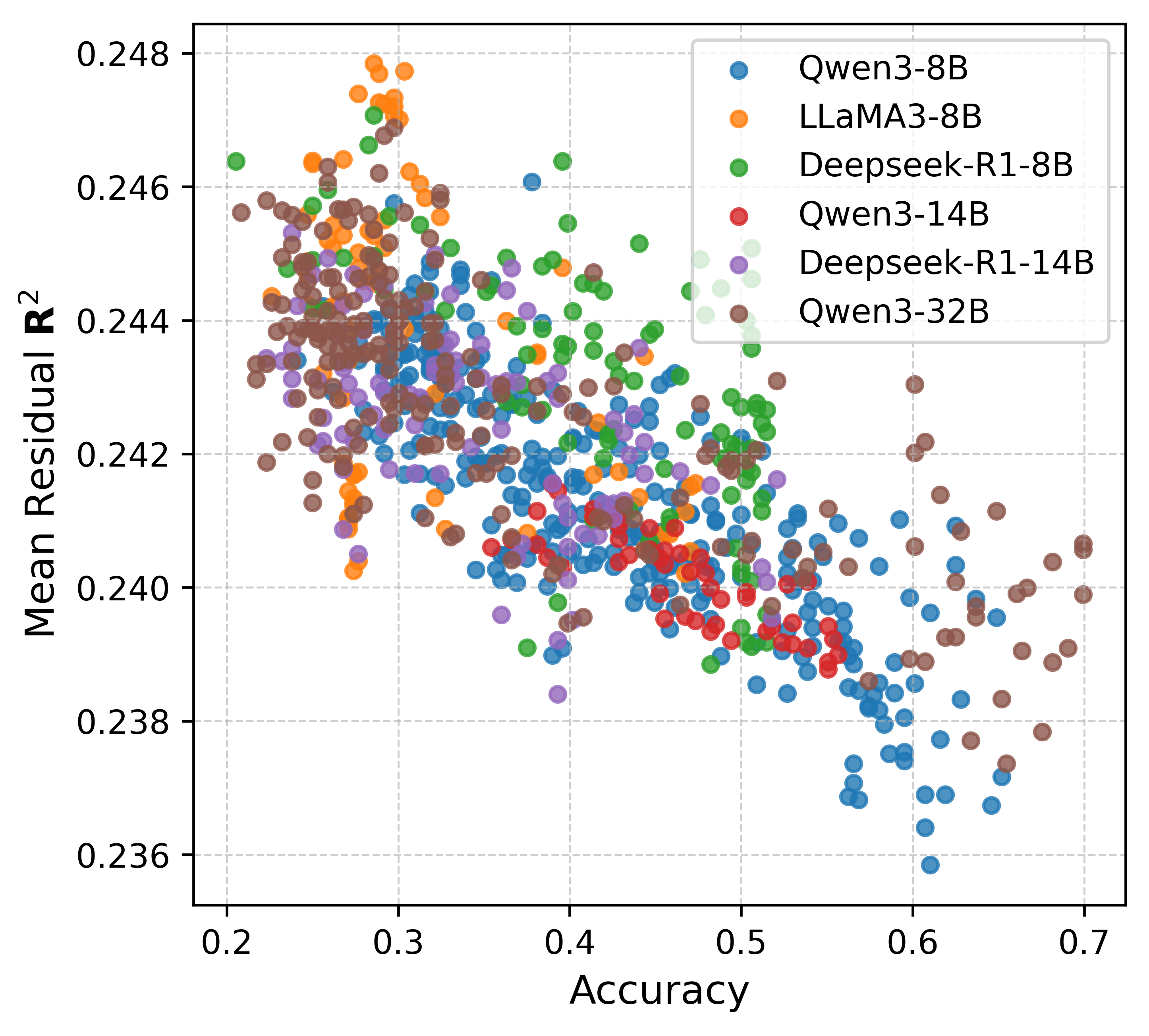}
    \caption{\textbf{GPQA}: Accuracy ($\uparrow$) vs residual risk $\mathbf{R}^2$ ($\downarrow$).}
    \label{fig:incomplete-gpqa}
  \end{subfigure}
  \caption{Performance of incomplete models across datasets. We report the results with 5 golden demonstrations and adapted last layer of model. Test accuracy increases as mean residual risk $R$ decreases, consistently across benchmarks (e.g., MMLU and GPQA), supporting our theory that prediction error is monotone in $\mathbf{R}^2$.}
  \label{fig:incomplete-models-duo}
\end{figure}

To validate our theory in quantifying the learning performance with respect to incomplete LLMs, we conduct experiments on incomplete models, i.e., $R(x,y)\neq 0$. For this purpose, we first fine-tune the LLMs with the last layer output, then record the output metrics. For each model and dataset, we measure two metrics: (i) accuracy: the standard evaluation metric on the benchmark given 5 golden demonstrations; (ii) the mean residual $\textbf{R}^2 = \sum_{x,y}P_X(x)R^2(x,y)$: the residual measure defined from our upper bound, which quantifies the amount of unexplained variability not captured by the concept bottleneck representation. A higher $\textbf{R}^2$ indicates stronger incompleteness in the model’s conceptual alignment with the task. Figure 3 reports the relationships between the label predicting accuracy and the residual $\textbf{R}^2$ across different models. A clear negative correlation emerges: (i) $\textbf{R}^2$ decreases as accuracy increases: Models achieving higher accuracy consistently exhibit smaller residual $\textbf{R}^2$, implying that their representations better align with the task concepts. (ii) Incomplete models suffer from high $\textbf{R}^2$: For weaker baselines such as LLaMA3-8B, the residual $\textbf{R}^2$ remains high, reflecting the incompleteness of the learnt knowledge; (iii) Generalization across families: The observed trend is consistent across different architectures and training paradigms, demonstrating that the proposed residual measure is not model-specific but captures a general phenomenon of incompleteness.

\section{Conclusion}
In this paper, we present theoretical analyses of CB-ICL, which reveals the fundamental mechanism of why and how ICL can perform well in prompts with only a few demonstrations. Moreover, our theory quantifies the knowledge leveraged by the pre-trained LLM embeddings, the similarity of the prompt demonstrations and query input text, as well as the impact of the number of prompt demonstrations and the dimensions of LLM embeddings, which provides useful guidance for model pre-training and prompt engineering. Finally, the effectiveness of our theory is validated by several real-data experiments.

\bibliography{iclr2026_conference}
\bibliographystyle{iclr2026_conference}

\appendix
\section{Proofs of Lemmas and Theorems}

\subsection{Proof of Lemma \ref{lem:lambda_1_ub}}\label{append:lemma1_proof}

In this section, we aim to prove Lemma \ref{lem:lambda_1_ub}. We separate the proof into three steps. First, we analyze the spectrum of $\textbf{A}(x_i)$ for a fixed $i$. After that, we derive the upper bound of the largest eigenvalue of $\textbf{A}(x_i)$. Finally, we extend the result to $\textbf{A}(x^n)$.

\begin{lemma}\label{lem:spectrum}
Let 
\[
\textbf{A}(x_i) = \mathsf{diag}\{P_{Y|X}(1|x_i),\dots,P_{Y|X}(M|x_i)\}-\phi_i\phi_i^\mathrm{T},
\]
where $\phi_i = [P_{Y|X}(1|x_i),\dots,P_{Y|X}(M|x_i)]^\mathrm{T}$. 
Then the eigenvalues of $\textbf{A}(x_i)$ are either some $P_{Y|X}(j|x_i)$ or solutions to
\begin{equation}\label{equ:eig-root}
1+\sum_{j=1}^M \frac{P_{Y|X}^2(j|x_i)}{\lambda - P_{Y|X}(j|x_i)}=0,
\end{equation}
with all such solutions satisfying $\lambda \leq \max_j P_{Y|X}(j|x_i)$.
\end{lemma}
\begin{proof}
The eigenvalues of $\textbf{A}(x_i)$ are the roots of the equation 
\begin{equation}\label{equ:spectrum_equ}
\det(\lambda I- \mathsf{diag}\left\{ P_{Y|X=x_i} \right\} + \phi_i\phi_i^\mathrm{T})=0, 
\end{equation}
where $\mathsf{diag}\left\{ P_{Y|X=x_i} \right\}=\mathsf{diag}\{P_{Y|X}(1|x_i),\dots,P_{Y|X}(M|x_i)\}$. 
For $\lambda\ne P_{Y|X}(j|x_i),\forall j$, the determinant can be factorized as
\begin{align*}
&\det(\lambda I- {\mathsf{diag}}\left\{ P_{Y|X=x_i} \right\} + \phi_i\phi_i^\mathrm{T})
\\
&= \det(\lambda I- {\mathsf{diag}}\left\{ P_{Y|X=x_i} \right\} )(1+ \phi_i^\mathrm{T}(\lambda I- {\rm diag}\left\{ P_{Y|X=x_i} \right\})^{-1}\phi_i)\\
&= \prod_{j=1}^M(\lambda-P_{Y|X}(j|x_i))\left(1+\sum_{j=1}^M\frac{P_{Y|X}^2(j|x_i)}{\lambda-P_{Y|X}(j|x_i)} \right)\end{align*}

Thus, eigenvalues are either some $P_{Y|X}(j|x_i)$ or solutions to Eq. (\ref{equ:eig-root}).
Furthermore, we have that all the solutions $\lambda$ to the Eq. (\ref{equ:eig-root}) must satisfy that $\lambda < \max_j P_{Y|X}(j|x_i)$, which makes all roots of Eq. (\ref{equ:spectrum_equ}) no larger than $\max_j P_{Y|X}(j|x_i)$.
\end{proof}

\begin{lemma}\label{lem:upper-bound}
For each $i$, the largest eigenvalue of $\textbf{A}(x_i)$ satisfies
\[
\lambda_1(\textbf{A}(x_i)) \leq 2 \max_j P_{Y|X}(j|x_i)\big(1-\max_j P_{Y|X}(j|x_i)\big).
\]
\end{lemma}
\begin{proof}
To prove this result, we distinguish between two cases: whether the maximum probability is attained by more than one label, or by a unique label.

\paragraph{Case 1: The size of set $\arg\max_j P_{Y|X}(j|x_i)$ is larger than $1$}
By continuity, we have 
\begin{equation*}
\det(P_{Y|X}(j|x_i)\cdot I- {\rm diag}\left\{ P_{Y|X=x_i} \right\} + \phi_i\phi_i^\mathrm{T})= P_{Y|X}^2(j|x_i)\prod_{k\ne j}(P_{Y|X}(j|x_i)-P_{Y|X}(k|x_i)).
\end{equation*}
Since $\arg\max_y P_{Y|X}(y|x_i)$ is not a single point set, $\max_y P_{Y|X}(y|x_i)$ is one solution to Eq. (\ref{equ:spectrum_equ}). Then, the largest eigenvalue of $\mathbf{A}(x_i)$ is $\max_y P_{Y|X}(y|x_i)$. In this case, since the size of set  $\arg\max_j P_{Y|X}(j|x_i)$ is larger than $1$, we have that $\max_j P_{Y|X}(j|x_i) \leq 1/2$, which gives that $\max_j P_{Y|X}(j|x_i) \leq 2\max_j P_{Y|X}(j|x_i)(1-\max_j P_{Y|X}(j|x_i))$.

\paragraph{Case 2: The size of set  $\arg\max_j P_{Y|X}(j|x_i)$ equals to $1$}
Let $j_0=\arg\max_j P_{Y|X}(j|x_i)$ and denote the largest root of the Eq. (\ref{equ:eig-root}) by $\lambda_0$.
Further, let arbitrary $j_1\in\argmax_{j,j\ne j_0} P_{Y|X}(j|x_i)$. Note that
\begin{equation*}
    1+\sum_{j=1}^M\frac{P_{Y|X}^2(j|x_i)}{\lambda-P_{Y|X}(j|x_i)}>1,\forall\lambda>P_{Y|X}(j_0|x_i),
\end{equation*}
\begin{equation*}
    \lim_{\lambda\to P_{Y|X}(j_0|x_i)^-}1+\sum_{j=1}^M\frac{P_{Y|X}^2(j|x_i)}{\lambda-P_{Y|X}(j|x_i)}=-\infty,
\end{equation*}
and
\begin{equation*}
    \lim_{\lambda\to P_{Y|X}(j_1|x_i)^+}1+\sum_{j=1}^M\frac{P_{Y|X}^2(j|x_i)}{\lambda-P_{Y|X}(j|x_i)}=\infty.
\end{equation*}
We have $P_{Y|X}(j_1|x_i) < \lambda_0 <P_{Y|X}(j_0|x_i)$. Following the analysis in Case 1, the value $P_{Y|X}(j_0|x_i)$ will not be the root of Eq. (\ref{equ:spectrum_equ}). Hence, we have the largest eigenvalue of $\mathbf{A}(x_i)$ must be $\lambda_0$.

In this case, we have the
\begin{equation*}
    \frac{P_{Y|X}^2(j_0|x_i)}{\lambda_0-P_{Y|X}(j_0|x_i)}+\frac{P_{Y|X}^2(j_1|x_i)}{\lambda_0-P_{Y|X}(j_1|x_i)}\le\overbrace{\sum_{j=1}^M\frac{P_{Y|X}^2(j|x_i)}{\lambda_0-P_{Y|X}(j|x_i)}}^{=-1}
\end{equation*}
and
\begin{align*}
    \overbrace{\sum_{j=1}^M\frac{P_{Y|X}^2(j|x_i)}{\lambda_0-P_{Y|X}(j|x_i)}}^{=-1} & \le\frac{P_{Y|X}^2(j_0|x_i)}{\lambda_0-P_{Y|X}(j_0|x_i)}+\frac{\sum_{k\ne j}P_{Y|X}^2(k|x_i)}{\lambda_0-P_{Y|X}(j_1|x_i)}\\
    & \le\frac{P_{Y|X}^2(j_0|x_i)}{\lambda_0-P_{Y|X}(j_0|x_i)}+\frac{(1-P_{Y|X}(j_0|x_i))^2}{\lambda_0-(1-P_{Y|X}(j_0|x_i))}
\end{align*}

Hence, solving 
$$
\frac{P_{Y|X}^2(j_0|x_i)}{\lambda_0-P_{Y|X}(j_0|x_i)}+\frac{(1-P_{Y|X}(j_0|x_i))^2}{\lambda_0-(1-P_{Y|X}(j_0|x_i))}\ge-1
$$
gives $\lambda_0\le2P_{Y|X}(j_0|x_i)(1-P_{Y|X}(j_0|x_i))$, where the equality is achieved by Bernoulli distribution.

In conclusion of the two cases, both cases give that
\begin{equation*}
    \lambda_1(\textbf{A}(x_i)) \leq 2\max_y P_{Y|X}(y|x_i)\left(1-\max_y P_{Y|X}(y|x_i)\right).
\end{equation*}
This completes the proof.
\end{proof}

\begin{lemma}\label{lem:block}
The largest eigenvalue of the block matrix $\textbf{A}(x^n)$ satisfies
\[
\lambda_1(\textbf{A}(x^n)) = \max_i \lambda_1(\textbf{A}(x_i)).
\]
\end{lemma}
\begin{proof}
By the Rayleigh–Ritz characterization,
\begin{equation*}
    \lambda_1(\textbf{A}(x^n)) = \max_{\|\bm{v}\|=1} \bm{v}^\text{T} \textbf{A}(x^n)\bm{v}.
\end{equation*}
Write $\bm{v}^\text{T} = [\bm{v}_1^\text{T},\dots, \bm{v}_n^\text{T}]$, with $\bm{v}_i$ that matches the size of $\textbf{A}(x_i)$. Then
\begin{equation*}
    \bm{v}^\text{T}\textbf{A}(x^n)\bm{v} = \sum_{i=1}^n \bm{v}_i^\text{T} \textbf{A}(x_i)\bm{v}_i \leq \sum_{i=1}^n\lambda_1(\textbf{A}(x_i)) \| \bm{v}_i\|^2 \leq \max_i \lambda_1(\textbf{A}(x_i)).
\end{equation*}
Let $j\in\argmax_i \lambda_1(\textbf{A}(x_i))$ and let $\bm{u}_j$ be a unit eigenvector of $\textbf{A}(x_j)$ for $\lambda_1(\textbf{A}(x_j))$. Define $\bm{v}$ by $\bm{v}_j=\bm{u}_j$ and $\bm{v}_i=\bm{0}$ for $i\neq j$. Then $\|\bm{v}\|=1$ and
\begin{equation*}
    \bm{v}^\text{T}\textbf{A}(x^n)\bm{v} = \bm{u}_j^\text{T}\textbf{A}(x_j)\bm{u}_j =\lambda_1(\textbf{A}(x_j))= \max_i \lambda_1(\textbf{A}(x_i)).
\end{equation*}
Therefore, $\lambda_1(\textbf{A}(x^n)) \geq \max_i \lambda_1(\textbf{A}(x_i)).$ Combining both inequalities yields
\begin{equation*}
    \lambda_1(\textbf{A}(x^n)) = \max_i \lambda_1(\textbf{A}(x_i)).
\end{equation*}
This completes the proof.
\end{proof}

By Lemmas \ref{lem:spectrum}, \ref{lem:upper-bound}, and \ref{lem:block}, we conclude the proof of Lemma \ref{lem:lambda_1_ub}.

\subsection{Proof of Theorem \ref{thm:compelete_sufficient}}\label{append:theorem1_proof}

To prove the Theorem \ref{thm:compelete_sufficient}, we first develop the lemma for expectation value of estimated concept $\underline{\hat{\alpha}}$.
\begin{lemma}\label{lem:alpha_mean}
For sufficient prompt demonstrations, i.e., $\mathbf{F}_n(x^n)$ is invertible, it holds that for all $x^n$,
\begin{equation*}
    \mathbb{E}_{P_{Y^n|X^n}}[\underline{\hat{\alpha}}(x^n,Y^n) | X^n=x^n] = \underline{\alpha}.
\end{equation*}
\end{lemma}
\begin{proof}
From definition of $\underline{\hat{\alpha}}$ in Eq. (\ref{equ:def_alpha}), we have that
\begin{align*}
\begin{aligned}
& \mathbb{E}_{P_{Y^n|X^n}}\left[ \underline{\hat{\alpha}}(x^n, Y^n) | X^n = x^n \right]\\ 
& = \mathbb{E}_{P_{Y^n|X^n}}\left[ \mathbf{F}_n^{-1}(x^n) \bar{f}_n(x^n,Y^n)  | X^n=x^n\right]\\
& = \sum_{y^n} \prod_{i=1}^n P_{Y\vert X}(y_i\vert x_i) \left(\frac{1}{n}\sum_{i=1}^n \sum_{y}  f(x_i,y)f^{\mathrm{T}}(x_i,y)\right)^{-1} \left(\frac{1}{n}\sum_{i=1}^nf(x_i,y_i)\right)\\
& = \left(\frac{1}{n}\sum_{i=1}^n \sum_{y}  f(x_i,y)f^{\mathrm{T}}(x_i,y)\right)^{-1}  \frac1n \sum_{i=1}^n \sum_y P_{Y\vert X}(y\vert x_i) f(x_i,y)\\
& = \left(\frac{1}{n}\sum_{i=1}^n \sum_{y}  f(x_i,y)f^{\mathrm{T}}(x_i,y)\right)^{-1} \frac1n\sum_{i=1}^n \sum_y f(x_i,y) f^\mathrm{T}(x_i,y) \underline{\alpha}\\
&=\underline{\alpha}.
\end{aligned}
\end{align*}
\end{proof}

With the Lemma \ref{lem:alpha_mean}, we directly have the written form of excessive risk for a complete model (i.e. $R(x,y)=0$ for all $x,y$).

\begin{corollary}\label{col:excessive_risk_complete}
For sufficient prompt demonstrations and complete model, the excessive risk can be written as
\begin{equation*}
    \mathbb{E}_{P_{Y^n|X^n}}\left[ 
\ell \left(x^n,Y^n;x_Q\right) | X^n = x^n\right] = \mathrm{tr}\left\{ \mathbf{F}_n^{-1}(x^n)\mathbf{F}(x_Q)\mathbf{F}_n^{-1}(x^n) \mathsf{F}(x^n) \right\}  - \underline{\alpha}^{\mathrm{T}} \mathbf{F}(x_Q)\underline{\alpha}, 
\end{equation*}
where $\mathsf{F}(x^n) = \mathbb{E}_{P_{Y^n|X^n}}\left[\bar{f}_n(x^n,Y^n) \bar{f}_n^\mathrm{T}(x^n,Y^n) | X^n=x^n\right]$.
\end{corollary}
\begin{proof}
    From the definition of $\ell \left(x^n,y^n;x_Q\right)$, we have that (note that $\mathbf{F}_n^{-1}(x^n)$ is symmetric)
\begin{align*}
&\mathbb{E}_{P_{Y^n|X^n}}\left[ 
\ell \left(x^n,Y^n;x_Q\right) | X^n = x^n\right]\\
& = \mathbb{E}_{P_{Y^n|X^n}}\left[ 
\sum_{y} \left(\hat{P}_{Y|X}(y|x_Q) -  P_{Y|X}(y|x_Q) \right)^2 | X^n = x^n\right]\\
& = \mathbb{E}_{P_{Y^n|X^n}}\left[ \sum_y \left( \underline{\hat{\alpha}}^\mathrm{T} f(x_Q,y) - \underline{\alpha}^\mathrm{T}f(x_Q,y) \right)^2 | X^n=x^n\right]\\
& = \mathbb{E}_{P_{Y^n|X^n}}\left[ \left(\underline{\hat{\alpha}} - \underline{\alpha}\right)^\mathrm{T} \mathbf{F}(x_Q) \left(\underline{\hat{\alpha}} - \underline{\alpha}\right) | X^n=x^n \right]\\
& = \mathbb{E}_{P_{Y^n|X^n}}\left[\underline{\hat{\alpha}}^{\mathrm{T}} \mathbf{F}(x_Q)\underline{\hat{\alpha}}| X^n=x^n\right] -2 \mathbb{E}_{P_{Y^n|X^n}}\left[\underline{\hat{\alpha}}^{\mathrm{T}} \mathbf{F}(x_Q)\underline{\alpha}| X^n=x^n\right] + \underline{\alpha}^{\mathrm{T}} \mathbf{F}(x_Q)\underline{\alpha} \\
& = \mathbb{E}_{P_{Y^n|X^n}}\left[\mathrm{tr}\left\{ \mathbf{F}_n^{-1}(x^n)^\mathrm{T}\mathbf{F}(x_Q)\mathbf{F}_n^{-1}(x^n) \bar{f}_n(x^n,Y^n) \bar{f}_n^\mathrm{T}(x^n,Y^n) \right\} | X^n=x^n\right] - \underline{\alpha}^{\mathrm{T}} \mathbf{F}(x_Q)\underline{\alpha}\\
& = \mathrm{tr}\left\{ \mathbf{F}_n^{-1}(x^n)\mathbf{F}(x_Q)\mathbf{F}_n^{-1}(x^n) \mathsf{F}(x^n) \right\}  - \underline{\alpha}^{\mathrm{T}} \mathbf{F}(x_Q)\underline{\alpha}.
\end{align*}
\end{proof}

Further, $\mathsf{F}(x^n)$ can be written as
\begin{align*}
    \mathsf{F}(x^n) = \frac1{n^2}f(x^n) \tilde{\mathbf{A}}(x^n)f^\mathrm{T}(x^n),
\end{align*}
where
\begin{equation*}
    \tilde{\mathbf{A}}(x^n) = \textsf{diag}\{\textsf{diag}\{P_{Y|X=x_1}\},\dots,\textsf{diag}\{P_{Y|X=x_n}\} \} + \sum_{i=1}^n\sum_{j=1}^n v_iv_j^\mathrm{T} - \sum_i v_iv_i^\mathrm{T},
\end{equation*}
with $\textsf{diag}\{P_{Y|X=x_i}\} = \textsf{diag}\{P_{Y|X}(1|x_i),\dots,P_{Y|X}(M|x_i)\}$, and
\begin{equation*}
    v_i = \left[0,\dots,0,P_{Y|X}(1|x_i),\dots, P_{Y|X}(M|x_i),0\dots,0\right]^\mathrm{T}.
\end{equation*}

In the next step, we show the connection between $v_i$, $f(x^n)$ and $\underline{\alpha}$.

\begin{lemma}\label{lem:aFa_comple}
It holds that
\begin{equation*}
    \frac1{n^2}\mathrm{tr}\left\{ \mathbf{F}_n^{-1}(x^n)\mathbf{F}(x_Q)\mathbf{F}_n^{-1}(x^n) f(x^n)\left(\sum_{i,j}v_iv_j^\mathrm{T}\right) f^\mathrm{T}(x^n)\right\}  = \underline{\alpha}^{\mathrm{T}} \mathbf{F}(x_Q)\underline{\alpha}.
\end{equation*}
\end{lemma}
\begin{proof}
We have that
\begin{align*}
& \frac1{n^2}\mathrm{tr}\left\{ \mathbf{F}_n^{-1}(x^n)\mathbf{F}(x_Q)\mathbf{F}_n^{-1}(x^n) f(x^n)\left(\sum_{i,j}v_iv_j^\mathrm{T}\right) f^\mathrm{T}(x^n)\right\}\\
& = \mathrm{tr}\left\{ \mathbf{F}_n^{-1}(x^n)\mathbf{F}(x_Q)\mathbf{F}_n^{-1}(x^n)\left(\sum_{i,j}\left(\frac1n f(x^n) v_i\right)\left(\frac1n f(x^n) v_j\right)^\mathrm{T}\right)\right\}\\
& = \mathrm{tr}\left\{ \mathbf{F}_n^{-1}(x^n)\mathbf{F}(x_Q)\mathbf{F}_n^{-1}(x^n)\left(\frac1n\sum_{i,y}P_{Y|X}(y|x_i) f(x_i, y)\right)\left(\frac1n\sum_{j,y}P_{Y|X}(y|x_j) f(x_j,y)\right)^\mathrm{T}\right\}\\
\end{align*}
From model \ref{equ:model}, we have that for a complete model, $P_{Y|X}(y|x) = \underline{\alpha}^\mathrm{T}f(x,y)$. Hence, we have that
\begin{align*}
& \mathrm{tr}\left\{ \mathbf{F}_n^{-1}(x^n)\mathbf{F}(x_Q)\mathbf{F}_n^{-1}(x^n)\left(\frac1n\sum_{i,y}P_{Y|X}(y|x_i) f(x_i, y)\right)\left(\frac1n\sum_{j,y} P_{Y|X}(y|x_j)f(x_j,y)\right)^\mathrm{T}\right\}\\
& = \mathrm{tr}\left\{\mathbf{F}_n^{-1}(x^n) \mathbf{F}(x_Q)\mathbf{F}_n^{-1}(x^n)\left(\frac1n\sum_{i,y} f(x_i, y)f^\mathrm{T}(x_i, y) \underline{\alpha} \right)\left(\frac1n\sum_{j,y} f(x_j,y)f^\mathrm{T}(x_j,y)\underline{\alpha}\right)^\mathrm{T}\right\}\\
& = \mathrm{tr}\left\{\mathbf{F}_n^{-1}(x^n) \mathbf{F}(x_Q)\mathbf{F}_n^{-1}(x^n)\mathbf{F}_n(x^n) \underline{\alpha} \left(\mathbf{F}_n(x^n)\underline{\alpha}\right)^\mathrm{T}\right\}\\
& = \underline{\alpha}^\mathrm{T}\mathbf{F}(x_Q) \underline{\alpha} .\\
\end{align*}
This completes the proof.
\end{proof}

Hence, we have the excessive risk without $\underline{\alpha}$.
\begin{corollary}\label{col:excessive_risk_complete_simple}
Let 
\[
\textbf{A}(x_i) = \mathsf{diag}\{P_{Y|X}(1|x_i),\dots,P_{Y|X}(M|x_i)\}-\phi_i\phi_i^\mathrm{T},
\]
where $\phi_i = [P_{Y|X}(1|x_i),\dots,P_{Y|X}(M|x_i)]^\mathrm{T}$, and let
\begin{equation*}
\textbf{A}(x^n) = \mathsf{diag}\left\{ \textbf{A}(x_1), \ldots, \textbf{A}(x_n) \right\},
\end{equation*}
it holds that
\begin{equation*}
    \mathbb{E}_{P_{Y^n|X^n}}\left[ 
\ell \left(x^n,Y^n;x_Q\right) | X^n = x^n\right]= \mathrm{tr}\left\{ \frac1{n^2}\mathbf{F}_n^{-1}(x^n)\mathbf{F}(x_Q)\mathbf{F}_n^{-1}(x^n) f(x^n) {\mathbf{A}}(x^n) f^\mathrm{T}(x^n) \right\}.
\end{equation*}
\end{corollary}
\begin{proof}
    Based on the Corollary \ref{col:excessive_risk_complete} and Lemma \ref{lem:aFa_comple}, we have the excessive risk being
\begin{align*}
& \mathbb{E}_{P_{Y^n|X^n}}\left[ 
\ell \left(x^n,Y^n;x_Q\right) | X^n = x^n\right]\\
&= \mathrm{tr}\left\{ \mathbf{F}_n^{-1}(x^n)\mathbf{F}(x_Q)\mathbf{F}_n^{-1}(x^n) \mathsf{F}(x^n) \right\}  - \underline{\alpha}^{\mathrm{T}} \mathbf{F}(x_Q)\underline{\alpha}\\
& = \mathrm{tr}\left\{ \frac1{n^2}\mathbf{F}_n^{-1}(x^n)\mathbf{F}(x_Q)\mathbf{F}_n^{-1}(x^n) f(x^n) \tilde{\mathbf{A}}(x^n)f^\mathrm{T}(x^n) \right\}  - \underline{\alpha}^{\mathrm{T}} \mathbf{F}(x_Q)\underline{\alpha}\\
& = \mathrm{tr}\left\{ \frac1{n^2}\mathbf{F}_n^{-1}(x^n)\mathbf{F}(x_Q)\mathbf{F}_n^{-1}(x^n) f(x^n) \left(\tilde{\mathbf{A}}(x^n) - \sum_{i,j}v_iv_j^\mathrm{T}\right)f^\mathrm{T}(x^n) \right\}.
\end{align*}
Note that $\tilde{\mathbf{A}}(x^n) = \sum_{i,j}v_iv_j^\mathrm{T} + \mathbf{A}(x^n)$, we have that
\begin{align*}
& \mathbb{E}_{P_{Y^n|X^n}}\left[ 
\ell \left(x^n,Y^n;x_Q\right) | X^n = x^n\right]= \mathrm{tr}\left\{ \frac1{n^2}\mathbf{F}_n^{-1}(x^n)\mathbf{F}(x_Q)\mathbf{F}_n^{-1}(x^n) f(x^n) {\mathbf{A}}(x^n) f^\mathrm{T}(x^n) \right\}.
\end{align*}
\end{proof}

Based on this form, we further find the upper bound.

\begin{lemma}\label{lem:better_ub_complete}
For sufficient prompt demonstrations and complete model, the excessive risk can be bounded by
\begin{equation*}
    \mathbb{E}_{P_{Y^n|X^n}}\left[ 
\ell \left(x^n,Y^n;x_Q\right) | X^n = x^n\right] \leq \frac1n \lambda_1\left(\mathbf{F}(x_Q)\mathbf{F}_n^{-1}(x^n) \right) S_K({\textbf{A}}(x^n)),
\end{equation*}
where $S_K({\textbf{A}}(x^n))$ is the sum of top K eigenvalues.
\end{lemma}
\begin{proof}
Denote $\mathbf{F}_n^{-1/2}(x^n)$ as $\mathbf{F}_n^{-1}(x^n) = \mathbf{F}_n^{-1/2}(x^n)\mathbf{F}_n^{-1/2}(x^n)$. Then, by Corollary \ref{col:excessive_risk_complete_simple}, the excessive risk can be written as
\begin{align*}
& \mathbb{E}_{P_{Y^n|X^n}}\left[ 
\ell \left(x^n,Y^n;x_Q\right) | X^n = x^n\right]\\
& = \mathrm{tr}\left\{ \frac1{n^2}\mathbf{F}_n^{-1/2}(x^n)\mathbf{F}(x_Q)\mathbf{F}_n^{-1/2}(x^n) \mathbf{F}_n^{-1/2}(x^n) f(x^n){\mathbf{A}}(x^n) f^\mathrm{T}(x^n)\mathbf{F}_n^{-1/2}(x^n) \right\}.
\end{align*}
Denote the $i$-th large eigenvalue for one matrix as $\lambda_i(\cdot)$. Note that both matrices $\mathbf{F}_n^{-1/2}(x^n)\mathbf{F}(x_Q)\mathbf{F}_n^{-1/2}(x^n)$ and $\mathbf{F}_n^{-1/2}(x^n) f(x^n){\mathbf{A}}(x^n) f^\mathrm{T}(x^n)\mathbf{F}_n^{-1/2}(x^n)$ are Hermitian. From Von Neumann's trace inequality \citep{mirsky1975trace}, we have the
\begin{align*}
& \mathrm{tr}\left\{ \frac1{n^2}\mathbf{F}_n^{-1/2}(x^n)\mathbf{F}(x_Q)\mathbf{F}_n^{-1/2}(x^n) \mathbf{F}_n^{-1/2}(x^n) f(x^n){\mathbf{A}}(x^n) f^\mathrm{T}(x^n)\mathbf{F}_n^{-1/2}(x^n) \right\}\\
& \leq \frac1n \sum_{i=1}^K \lambda_i\left(\mathbf{F}_n^{-1/2}(x^n)\mathbf{F}(x_Q)\mathbf{F}_n^{-1/2}(x^n)\right) \lambda_i\left(\frac1{\sqrt{n}}\mathbf{F}_n^{-1/2}(x^n) f(x^n){\mathbf{A}}(x^n) f^\mathrm{T}(x^n)\mathbf{F}_n^{-1/2}(x^n)\frac1{\sqrt{n}}\right).
\end{align*}
Since that
\begin{align*}
& \left(\frac1{\sqrt{n}}\mathbf{F}_n^{-1/2}(x^n) f(x^n)\right)\left(\frac1{\sqrt{n}}\mathbf{F}_n^{-1/2}(x^n) f(x^n)\right)^\mathrm{T}\\
& = \frac1{\sqrt{n}}\mathbf{F}_n^{-1/2}(x^n) f(x^n)f^\mathrm{T}(x^n)\mathbf{F}_n^{-1/2}(x^n)\frac1{\sqrt{n}}\\
& = \mathbf{F}_n^{-1/2}(x^n)\mathbf{F}_n(x^n)\mathbf{F}_n^{-1/2}(x^n)\\
& = I,
\end{align*}
the matrix $\frac1{\sqrt{n}}\mathbf{F}_n^{-1/2}(x^n) f(x^n)$ is an orthogonal matrix. Then, the eigenvalues of matrix $\frac1{\sqrt{n}}\mathbf{F}_n^{-1/2}(x^n) f(x^n){\mathbf{A}}(x^n) f^\mathrm{T}(x^n)\mathbf{F}_n^{-1/2}(x^n)\frac1{\sqrt{n}}$ equal to those of $\mathbf{A}(x^n)$, i.e.,
\begin{equation*}
\lambda_i\left(\frac1{\sqrt{n}}\mathbf{F}_n^{-1/2}(x^n) f(x^n){\mathbf{A}}(x^n) f^\mathrm{T}(x^n)\mathbf{F}_n^{-1/2}(x^n)\frac1{\sqrt{n}}\right) = \lambda_i\left(\mathbf{A}(x^n)\right).
\end{equation*}

Furthermore, we prove that $\mathbf{A}(x^n)$ is semi-positive definite.
For any vector $v=[v_1^\mathrm{T},\dots,v_n^\mathrm{T}]^\mathrm{T}$, we have that
\begin{align*}
& v^\mathrm{T}\mathbf{A}(x^n)v = \sum_{i=1}^n v_i^\mathrm{T} \mathbf{A}(x_i)v_i = \sum_{i=1}^n\sum_{j=1}^M P_{Y|X}(j|x_i)v^2_i(j) - \sum_{i=1}^n\left(\sum_{j=1}^M P_{Y|X}(j|x_i) v_i(j)\right)^2 \geq 0,
\end{align*}
where the inequality is achieved by Jensen's inequality.
Hence, all eigenvalues of $\mathbf{A}(x^n)$ are non-negative, making that
\begin{align*}
& \mathrm{tr}\left\{ \frac1{n^2}\mathbf{F}_n^{-1/2}(x^n)\mathbf{F}(x_Q)\mathbf{F}_n^{-1/2}(x^n) \mathbf{F}_n^{-1/2}(x^n) f(x^n){\mathbf{A}}(x^n) f^\mathrm{T}(x^n)\mathbf{F}_n^{-1/2}(x^n) \right\}\\
& \leq \frac1n \sum_{i=1}^K \lambda_i\left(\mathbf{F}_n^{-1/2}(x^n)\mathbf{F}(x_Q)\mathbf{F}_n^{-1/2}(x^n)\right) \lambda_i\left(\frac1{\sqrt{n}}\mathbf{F}_n^{-1/2}(x^n) f(x^n){\mathbf{A}}(x^n) f^\mathrm{T}(x^n)\mathbf{F}_n^{-1/2}(x^n)\frac1{\sqrt{n}}\right)\\
& \leq \frac1n \lambda_1\left(\mathbf{F}_n^{-1/2}(x^n)\mathbf{F}(x_Q)\mathbf{F}_n^{-1/2}(x^n)\right) \sum_{i=1}^K \lambda_i\left(\mathbf{A}(x^n)\right)\\
& = \frac1n \lambda_1\left(\mathbf{F}(x_Q)\mathbf{F}_n^{-1}(x^n) \right) S_K({\textbf{A}}(x^n)).
\end{align*}
\end{proof}

Based on Lemma \ref{lem:better_ub_complete}, we prove the Theorem \ref{thm:compelete_sufficient} by 
\begin{equation*}
    S_K({\textbf{A}}(x^n)) \leq K \lambda_1(\textbf{A}(x^n)).
\end{equation*}

\subsection{Proof of Theorem \ref{thm:complete_insufficient}}\label{append:theorem2_proof}
Similar to Lemma \ref{lem:alpha_mean}, we have the expectation of $\underline{\alpha}$ under insufficient demonstrations.
\begin{lemma}\label{lem:alpha_mean_insufficient}
For insufficient prompt demonstrations, i.e., $\mathbf{F}_n(x^n)$ is not invertible, it holds that for all $x^n$,
\begin{equation*}
    \mathbb{E}_{P_{Y^n|X^n}}[\underline{\hat{\alpha}}(x^n,Y^n) | X^n=x^n] = \mathbf{F}_n^{\dagger}(x^n)\mathbf{F}_n(x^n)\underline{\alpha}.
\end{equation*}
\end{lemma}
\begin{proof}
From definition of $\underline{\hat{\alpha}}$ in Eq. (\ref{equ:def_alpha}), we have that
\begin{align*}
\begin{aligned}
& \mathbb{E}_{P_{Y^n|X^n}}\left[ \underline{\hat{\alpha}}(x^n, Y^n) | X^n = x^n \right]\\ 
& = \mathbb{E}_{P_{Y^n|X^n}}\left[ \mathbf{F}_n^{\dagger}(x^n) \bar{f}_n(x^n,Y^n)  | X^n=x^n\right]\\
& = \sum_{y^n} \prod_{i=1}^n P_{Y\vert X}(y_i\vert x_i) \mathbf{F}_n^{\dagger}(x^n) \left(\frac{1}{n}\sum_{i=1}^nf(x_i,y_i)\right)\\
& = \mathbf{F}_n^{\dagger}(x^n)  \frac1n \sum_{i=1}^n \sum_y P_{Y\vert X}(y\vert x_i) f(x_i,y)\\
& = \mathbf{F}_n^{\dagger}(x^n) \frac1n\sum_{i=1}^n \sum_y f(x_i,y) f^\mathrm{T}(x_i,y) \underline{\alpha} \\
&=\mathbf{F}_n^{\dagger}(x^n)\mathbf{F}_n(x^n)\underline{\alpha}.
\end{aligned}
\end{align*}
\end{proof}

\begin{corollary}\label{col:excessive_risk_insufficient}
For insufficient prompt demonstrations and complete model, the excessive risk can be written as
\begin{align*}
& \mathbb{E}_{P_{Y^n|X^n}}\left[ 
\ell \left(x^n,Y^n;x_Q\right) | X^n = x^n\right]\\
& = \mathrm{tr}\left\{ \mathbf{F}_n^{\dagger}(x^n)\mathbf{F}(x_Q)\mathbf{F}_n^{\dagger}(x^n) \mathsf{F}(x^n) \right\} - 2\underline{\alpha}^\mathrm{T}\mathbf{F}(x_Q)\mathbf{F}_n^{\dagger}(x^n)\mathbf{F}_n(x^n)\underline{\alpha} + \underline{\alpha}^{\mathrm{T}} \mathbf{F}(x_Q)\underline{\alpha}, 
\end{align*}
where $\mathsf{F}(x^n) = \mathbb{E}_{P_{Y^n|X^n}}\left[\bar{f}_n(x^n,Y^n) \bar{f}_n^\mathrm{T}(x^n,Y^n) | X^n=x^n\right]$.
\end{corollary}
\begin{proof}
    From the definition of $\ell \left(x^n,y^n;x_Q\right)$, we have that (note that $\mathbf{F}_n^{-1}(x^n)$ is symmetric)
\begin{align*}
&\mathbb{E}_{P_{Y^n|X^n}}\left[ 
\ell \left(x^n,Y^n;x_Q\right) | X^n = x^n\right]\\
& = \mathbb{E}_{P_{Y^n|X^n}}\left[ 
\sum_{y} \left(\hat{P}_{Y|X}(y|x_Q) -  P_{Y|X}(y|x_Q) \right)^2 | X^n = x^n\right]\\
& = \mathbb{E}_{P_{Y^n|X^n}}\left[ \sum_y \left( \underline{\hat{\alpha}}^\mathrm{T} f(x_Q,y) - \underline{\alpha}^\mathrm{T}f(x_Q,y) \right)^2 | X^n=x^n\right]\\
& = \mathbb{E}_{P_{Y^n|X^n}}\left[ \left(\underline{\hat{\alpha}} - \underline{\alpha}\right)^\mathrm{T} \mathbf{F}(x_Q) \left(\underline{\hat{\alpha}} - \underline{\alpha}\right) | X^n=x^n \right]\\
& = \mathbb{E}_{P_{Y^n|X^n}}\left[\underline{\hat{\alpha}}^{\mathrm{T}} \mathbf{F}(x_Q)\underline{\hat{\alpha}}| X^n=x^n\right] -2 \mathbb{E}_{P_{Y^n|X^n}}\left[\underline{\hat{\alpha}}^{\mathrm{T}} \mathbf{F}(x_Q)\underline{\alpha}| X^n=x^n\right] + \underline{\alpha}^{\mathrm{T}} \mathbf{F}(x_Q)\underline{\alpha} \\
& = \mathrm{tr}\left\{ \mathbf{F}_n^{\dagger}(x^n)\mathbf{F}(x_Q)\mathbf{F}_n^{\dagger}(x^n) \mathsf{F}(x^n) \right\}  - 2\underline{\alpha}^\mathrm{T}\mathbf{F}(x_Q)\mathbf{F}_n^{\dagger}(x^n)\mathbf{F}_n(x^n)\underline{\alpha} + \underline{\alpha}^{\mathrm{T}} \mathbf{F}(x_Q)\underline{\alpha}.
\end{align*}
\end{proof}

Write $\mathsf{F}(x^n)$ as
\begin{align*}
    \mathsf{F}(x^n) = \frac1{n^2}f(x^n) \tilde{\mathbf{A}}(x^n)f^\mathrm{T}(x^n),
\end{align*}
where
\begin{equation*}
    \tilde{\mathbf{A}}(x^n) = \textsf{diag}\{\textsf{diag}\{P_{Y|X=x_1}\},\dots,\textsf{diag}\{P_{Y|X=x_n}\} \} + \sum_{i=1}^n\sum_{j=1}^n v_iv_j^\mathrm{T} - \sum_i v_iv_i^\mathrm{T},
\end{equation*}
with $\textsf{diag}\{P_{Y|X=x_i}\} = \textsf{diag}\{P_{Y|X}(1|x_i),\dots,P_{Y|X}(M|x_i)\}$, and
\begin{equation*}
    v_i = \left[0,\dots,0,P_{Y|X}(1|x_i),\dots, P_{Y|X}(M|x_i),0\dots,0\right]^\mathrm{T}.
\end{equation*}

\begin{lemma}\label{lem:aFa_comple_insufficient}
For an insufficient demonstrations and complete model, it holds that
\begin{align*}
& \frac1{n^2}\mathrm{tr}\left\{ \mathbf{F}_n^{\dagger}(x^n)\mathbf{F}(x_Q)\mathbf{F}_n^{\dagger}(x^n) f(x^n)\left(\sum_{i,j}v_iv_j^\mathrm{T}\right) f^\mathrm{T}(x^n)\right\}\\
& = \left(\mathbf{F}_n^{\dagger}(x^n)\mathbf{F}_n(x^n)\underline{\alpha}\right)^\mathrm{T}\mathbf{F}(x_Q) \mathbf{F}_n^{\dagger}(x^n)\mathbf{F}_n(x^n)\underline{\alpha}.
\end{align*}
\end{lemma}
\begin{proof}
We have that
\begin{align*}
& \frac1{n^2}\mathrm{tr}\left\{ \mathbf{F}_n^{\dagger}(x^n)\mathbf{F}(x_Q)\mathbf{F}_n^{\dagger}(x^n) f(x^n)\left(\sum_{i,j}v_iv_j^\mathrm{T}\right) f^\mathrm{T}(x^n)\right\}\\
& = \mathrm{tr}\left\{ \mathbf{F}_n^{\dagger}(x^n)\mathbf{F}(x_Q)\mathbf{F}_n^{\dagger}(x^n)\left(\frac1n\sum_{i,y}P_{Y|X}(y|x_i) f(x_i, y)\right)\left(\frac1n\sum_{j,y} P_{Y|X}(y|x_j)f(x_j,y)\right)^\mathrm{T}\right\}\\
\end{align*}
From model \ref{equ:model}, we have that for a complete model, $P_{Y|X}(y|x) = \underline{\alpha}^\mathrm{T}f(x,y)$. Hence, we have that
\begin{align*}
& \mathrm{tr}\left\{ \mathbf{F}_n^{\dagger}(x^n)\mathbf{F}(x_Q)\mathbf{F}_n^{\dagger}(x^n)\left(\frac1n\sum_{i,y}P_{Y|X}(y|x_i) f(x_i, y)\right)\left(\frac1n\sum_{j,y} P_{Y|X}(y|x_j) f(x_j,y)\right)^\mathrm{T}\right\}\\
& = \mathrm{tr}\left\{\mathbf{F}_n^{\dagger}(x^n) \mathbf{F}(x_Q)\mathbf{F}_n^{\dagger}(x^n)\left(\frac1n\sum_{i,y} f(x_i, y)f^\mathrm{T}(x_i, y) \underline{\alpha} \right)\left(\frac1n\sum_{j,y} f(x_j,y)f^\mathrm{T}(x_j,y)\underline{\alpha}\right)^\mathrm{T}\right\}\\
& = \mathrm{tr}\left\{\mathbf{F}_n^{\dagger}(x^n) \mathbf{F}(x_Q)\mathbf{F}_n^{\dagger}(x^n)\mathbf{F}_n(x^n) \underline{\alpha} \left(\mathbf{F}_n(x^n)\underline{\alpha}\right)^\mathrm{T}\right\}\\
& = \left(\mathbf{F}_n^{\dagger}(x^n)\mathbf{F}_n(x^n)\underline{\alpha}\right)^\mathrm{T}\mathbf{F}(x_Q) \mathbf{F}_n^{\dagger}(x^n)\mathbf{F}_n(x^n)\underline{\alpha}.\\
\end{align*}
This completes the proof.
\end{proof}

Hence, we have the excessive risk in a simplified form.
\begin{corollary}\label{col:excessive_risk_complete_insufficient_simple}
Let 
\[
\textbf{A}(x_i) = \mathsf{diag}\{P_{Y|X}(1|x_i),\dots,P_{Y|X}(M|x_i)\}-\phi_i\phi_i^\mathrm{T},
\]
where $\phi_i = [P_{Y|X}(1|x_i),\dots,P_{Y|X}(M|x_i)]^\mathrm{T}$, and let
\begin{equation*}
\textbf{A}(x^n) = \mathsf{diag}\left\{ \textbf{A}(x_1), \ldots, \textbf{A}(x_n) \right\},
\end{equation*}
it holds that
\begin{align*}
& \mathbb{E}_{P_{Y^n|X^n}}\left[ 
\ell \left(x^n,Y^n;x_Q\right) | X^n = x^n\right]\\
& = \mathrm{tr}\left\{ \frac1{n^2}\mathbf{F}_n^{\dagger}(x^n)\mathbf{F}(x_Q)\mathbf{F}_n^{\dagger}(x^n) f(x^n) {\mathbf{A}}(x^n) f^\mathrm{T}(x^n) \right\} +
\left\| f ^{\mathrm{T}}(x_Q)\mathbf{F}_n^\perp(x^n)\underline{\alpha}\right\|^2 ,
\end{align*}
where $\mathbf{F}_n^{\perp}(x^n) = \textbf{I}_K-\mathbf{F}_n^\dagger(x^n) \mathbf{F}_n(x^n)$.
\end{corollary}
\begin{proof}
    Based on the Corollary \ref{col:excessive_risk_complete} and Lemma \ref{lem:aFa_comple}, we have the excessive risk being
\begin{align*}
& \mathbb{E}_{P_{Y^n|X^n}}\left[ 
\ell \left(x^n,Y^n;x_Q\right) | X^n = x^n\right]\\
&= \mathrm{tr}\left\{ \mathbf{F}_n^{\dagger}(x^n)\mathbf{F}(x_Q)\mathbf{F}_n^{\dagger}(x^n) \mathsf{F}(x^n) \right\} - 2\underline{\alpha}^\mathrm{T}\mathbf{F}(x_Q)\mathbf{F}_n^{\dagger}(x^n)\mathbf{F}_n(x^n)\underline{\alpha} + \underline{\alpha}^{\mathrm{T}} \mathbf{F}(x_Q)\underline{\alpha}.\\
& = \mathrm{tr}\left\{ \frac1{n^2}\mathbf{F}_n^{-1}(x^n)\mathbf{F}(x_Q)\mathbf{F}_n^{-1}(x^n) f(x^n) \left(\tilde{\mathbf{A}}(x^n) - \sum_{i,j}v_iv_j^\mathrm{T}\right)f^\mathrm{T}(x^n) \right\}\\
& \quad + \left(\mathbf{F}_n^{\dagger}(x^n)\mathbf{F}_n(x^n)\underline{\alpha}\right)^\mathrm{T}\mathbf{F}(x_Q)\mathbf{F}_n^{\dagger}(x^n)\mathbf{F}_n(x^n)\underline{\alpha} - 2\underline{\alpha}^\mathrm{T}\mathbf{F}(x_Q)\mathbf{F}_n^{\dagger}(x^n)\mathbf{F}_n(x^n)\underline{\alpha} + \underline{\alpha}^{\mathrm{T}} \mathbf{F}(x_Q)\underline{\alpha}..
\end{align*}
Note that $\tilde{\mathbf{A}}(x^n) = \sum_{i,j}v_iv_j^\mathrm{T} + \mathbf{A}(x^n)$ and $\mathbf{F}(x_Q) = f(x_Q)f^\mathrm{T}(x_Q)$, we have that
\begin{align*}
& \mathbb{E}_{P_{Y^n|X^n}}\left[ 
\ell \left(x^n,Y^n;x_Q\right) | X^n = x^n\right]\\
& = \mathrm{tr}\left\{ \frac1{n^2}\mathbf{F}_n^{\dagger}(x^n)\mathbf{F}(x_Q)\mathbf{F}_n^{\dagger}(x^n) f(x^n) {\mathbf{A}}(x^n) f^\mathrm{T}(x^n) \right\} +
\left\| f ^{\mathrm{T}}(x_Q)\mathbf{F}_n^\perp(x^n)\underline{\alpha}\right\|^2.
\end{align*}
\end{proof}

\begin{lemma}\label{lem:better_ub_complete_insufficient}
For sufficient prompt demonstrations and complete model, the excessive risk can is bounded by
\begin{equation*}
    \mathbb{E}_{P_{Y^n|X^n}}\left[ 
\ell \left(x^n,Y^n;x_Q\right) | X^n = x^n\right] \leq \frac1n \lambda_1\left(\mathbf{F}(x_Q)\mathbf{F}_n^{\dagger}(x^n) \right) S_K({\textbf{A}}(x^n))+ \left\| f ^{\mathrm{T}}(x_Q)\mathbf{F}_n^\perp(x^n)\underline{\alpha}\right\|^2,
\end{equation*}
where $S_K({\textbf{A}}(x^n))$ is the sum of top K eigenvalues.
\end{lemma}
\begin{proof}
Denote $\mathbf{F}_n^\dagger(x^n)^{1/2}$ as $\mathbf{F}_n^{\dagger}(x^n) = \mathbf{F}_n^\dagger(x^n)^{1/2}\mathbf{F}_n^\dagger(x^n)^{1/2}$. Then, by Corollary \ref{col:excessive_risk_complete_insufficient_simple}, the excessive risk can be written as
\begin{align*}
& \mathbb{E}_{P_{Y^n|X^n}}\left[ 
\ell \left(x^n,Y^n;x_Q\right) | X^n = x^n\right]\\
& = \mathrm{tr}\left\{ \frac1{n^2}\mathbf{F}_n^\dagger(x^n)^{1/2}(x^n)\mathbf{F}(x_Q)\mathbf{F}_n^\dagger(x^n)^{1/2}\mathbf{F}_n^\dagger(x^n)^{1/2} f(x^n) {\mathbf{A}}(x^n) f^\mathrm{T}(x^n)\mathbf{F}_n^\dagger(x^n)^{1/2} \right\}\\
&\quad + \left\| f ^{\mathrm{T}}(x_Q)\mathbf{F}_n^\perp(x^n)\underline{\alpha}\right\|^2.
\end{align*}
From Von Neumann's trace inequality \citep{mirsky1975trace}, denote the $i$-th large eigenvalue for one matrix as $\lambda_i(\cdot)$, since both matrices $\mathbf{F}_n^\dagger(x^n)^{1/2}\mathbf{F}(x_Q)\mathbf{F}_n^\dagger(x^n)^{1/2}$ and $\mathbf{F}_n^\dagger(x^n)^{1/2}(x^n) f(x^n){\mathbf{A}}(x^n) f^\mathrm{T}(x^n)\mathbf{F}_n^\dagger(x^n)^{1/2}$ are Hermitian, we have that
\begin{align*}
& \mathrm{tr}\left\{ \frac1{n^2}\mathbf{F}_n^\dagger(x^n)^{1/2}\mathbf{F}(x_Q)\mathbf{F}_n^\dagger(x^n)^{1/2} \mathbf{F}_n^\dagger(x^n)^{1/2}f(x^n){\mathbf{A}}(x^n) f^\mathrm{T}(x^n)\mathbf{F}_n^\dagger(x^n)^{1/2} \right\}\\
& \leq \frac1n \sum_{i=1}^K \lambda_i\left(\mathbf{F}_n^\dagger(x^n)^{1/2}\mathbf{F}(x_Q)\mathbf{F}_n^\dagger(x^n)^{1/2}\right) \lambda_i\left(\frac1{\sqrt{n}}\mathbf{F}_n^\dagger(x^n)^{1/2} f(x^n){\mathbf{A}}(x^n) f^\mathrm{T}(x^n)\mathbf{F}_n^\dagger(x^n)^{1/2}\frac1{\sqrt{n}}\right).
\end{align*}
Since that
\begin{align*}
& \left(\frac1{\sqrt{n}}\mathbf{F}_n^\dagger(x^n)^{1/2} f(x^n)\right)\left(\frac1{\sqrt{n}}\mathbf{F}_n^\dagger(x^n)^{1/2} f(x^n)\right)^\mathrm{T}\\
& = \frac1{\sqrt{n}}\mathbf{F}_n^\dagger(x^n)^{1/2} f(x^n)f^\mathrm{T}(x^n)\mathbf{F}_n^\dagger(x^n)^{1/2}\frac1{\sqrt{n}}\\
& = \mathbf{F}_n^\dagger(x^n)^{1/2}\mathbf{F}_n(x^n)\mathbf{F}_n^\dagger(x^n)^{1/2}\\
\end{align*}
is a projection matrix, the eigenvalues of matrix $\frac1{\sqrt{n}}\mathbf{F}_n^\dagger(x^n)^{1/2} f(x^n){\mathbf{A}}(x^n) f^\mathrm{T}(x^n)\mathbf{F}_n^\dagger(x^n)^{1/2}\frac1{\sqrt{n}}$ equals to that of $\mathbf{A}(x^n)$, i.e.,
\begin{equation*}
\lambda_i\left(\frac1{\sqrt{n}}\mathbf{F}_n^\dagger(x^n)^{1/2} f(x^n){\mathbf{A}}(x^n) f^\mathrm{T}(x^n)\mathbf{F}_n^\dagger(x^n)^{1/2}\frac1{\sqrt{n}}\right) = \lambda_i\left(\mathbf{A}(x^n)\right).
\end{equation*}

As is proved in Lemma \ref{lem:better_ub_complete}, all eigenvalues of $\mathbf{A}(x^n)$ are non-nagetive, makeing that
\begin{align*}
& \mathrm{tr}\left\{ \frac1{n^2}\mathbf{F}_n^\dagger(x^n)^{1/2}\mathbf{F}(x_Q)\mathbf{F}_n^\dagger(x^n)^{1/2} \mathbf{F}_n^\dagger(x^n)^{1/2} f(x^n){\mathbf{A}}(x^n) f^\mathrm{T}(x^n)\mathbf{F}_n^\dagger(x^n)^{1/2} \right\}\\
& \leq \frac1n \sum_{i=1}^K \lambda_i\left(\mathbf{F}_n^\dagger(x^n)^{1/2}\mathbf{F}(x_Q)\mathbf{F}_n^\dagger(x^n)^{1/2}\right) \lambda_i\left(\frac1{\sqrt{n}}\mathbf{F}_n^\dagger(x^n)^{1/2} f(x^n){\mathbf{A}}(x^n) f^\mathrm{T}(x^n)\mathbf{F}_n^\dagger(x^n)^{1/2}\frac1{\sqrt{n}}\right)\\
& \leq \frac1n \lambda_1\left(\mathbf{F}_n^\dagger(x^n)^{1/2}\mathbf{F}(x_Q)\mathbf{F}_n^\dagger(x^n)^{1/2}\right) \sum_{i=1}^K \lambda_i\left(\mathbf{A}(x^n)\right)\\
& = \frac1n \lambda_1\left(\mathbf{F}(x_Q)\mathbf{F}_n^{\dagger}(x^n) \right) S_K({\textbf{A}}(x^n)).
\end{align*}
\end{proof}

Based on Lemma \ref{lem:better_ub_complete_insufficient}, we simply prove the Theorem \ref{thm:complete_insufficient} by 
\begin{equation*}
    S_K({\textbf{A}}(x^n)) \leq K \lambda_1(\textbf{A}(x^n)).
\end{equation*}

\subsection{Proof of Theorem \ref{thm:incomplete_insufficient}}\label{append:theorem3_proof}

Similar to Lemma \ref{lem:alpha_mean_insufficient}, we have the expectation of $\underline{\alpha}$ under insufficient demonstrations with incomplete model.
\begin{lemma}\label{lem:alpha_mean_incomplete_insufficient}
For insufficient prompt demonstrations, i.e., $\mathbf{F}_n(x^n)$ is not invertible, and incomplete model, i.e. $R(x,y)\neq 0$, it holds that for all $x^n$,
\begin{equation*}
    \mathbb{E}_{P_{Y^n|X^n}}[\underline{\hat{\alpha}}(x^n,Y^n) | X^n=x^n] = \mathbf{F}_n^{\dagger}(x^n)\mathbf{F}_n(x^n)\underline{\alpha} + \frac1n \mathbf{F}_n^{\dagger}(x^n) f(x^n) \mathsf{R}(x^n).
\end{equation*}
\end{lemma}
\begin{proof}
From definition of $\underline{\hat{\alpha}}$ in Eq. (\ref{equ:def_alpha}), we have that
\begin{align*}
\begin{aligned}
& \mathbb{E}_{P_{Y^n|X^n}}\left[ \underline{\hat{\alpha}}(x^n, Y^n) | X^n = x^n \right]\\ 
& = \mathbb{E}_{P_{Y^n|X^n}}\left[ \mathbf{F}_n^{\dagger}(x^n) \bar{f}_n(x^n,Y^n)  | X^n=x^n\right]\\
& = \mathbf{F}_n^{\dagger}(x^n)  \frac1n \sum_{i=1}^n \sum_y P_{Y\vert X}(y\vert x_i) f(x_i,y)\\
& = \mathbf{F}_n^{\dagger}(x^n) \frac1n\sum_{i=1}^n \sum_y f(x_i,y) f^\mathrm{T}(x_i,y) \underline{\alpha} + \mathbf{F}_n^{\dagger}(x^n)\frac1n\sum_{i=1}^n\sum_y f(x_i,y)R(x_i,y)\\
&=\mathbf{F}_n^{\dagger}(x^n)\mathbf{F}_n(x^n)\underline{\alpha}+ \frac1n \mathbf{F}_n^{\dagger}(x^n) f(x^n) \mathsf{R}(x^n).
\end{aligned}
\end{align*}
\end{proof}

\begin{corollary}\label{col:excessive_risk_incomplete_insufficient}
For insufficient prompt demonstrations and incomplete model, the excessive risk can be written as
\begin{align*}
& \mathbb{E}_{P_{Y^n|X^n}}\left[ 
\ell \left(x^n,Y^n;x_Q\right) | X^n = x^n\right]\\
& = \mathrm{tr}\left\{ \mathbf{F}_n^{\dagger}(x^n)\mathbf{F}(x_Q)\mathbf{F}_n^{\dagger}(x^n) \mathsf{F}(x^n) \right\}  - 2\underline{\alpha}^\mathrm{T}\mathbf{F}(x_Q)\mathbf{F}_n^{\dagger}(x^n)\mathbf{F}_n(x^n)\underline{\alpha} + \underline{\alpha}^{\mathrm{T}} \mathbf{F}(x_Q)\underline{\alpha}\\
& \quad -\frac2n \mathsf{R}^\mathrm{T}(x^n)f^\mathrm{T}(x^n)\mathbf{F}_n^{\dagger}(x^n)\mathbf{F}(x_Q)\underline{\alpha} - 2 \mathsf{R}^\mathrm{T}(x_Q)f^\mathrm{T}(x_Q)\mathbf{F}_n^\perp(x^n)\underline{\alpha}\\  
& \quad + \frac2n \mathsf{R}^\mathrm{T}(x_Q)f^\mathrm{T}(x_Q)\mathbf{F}_n^{\dagger}(x^n)f(x^n)\mathsf{R}(x^n) + \sum_y R^2(x_Q,y), 
\end{align*}
where $\mathsf{F}(x^n) = \mathbb{E}_{P_{Y^n|X^n}}\left[\bar{f}_n(x^n,Y^n) \bar{f}_n^\mathrm{T}(x^n,Y^n) | X^n=x^n\right]$.
\end{corollary}
\begin{proof}
    From the definition of $\ell \left(x^n,y^n;x_Q\right)$, we have that (note that $\mathbf{F}_n^{-1}(x^n)$ is symmetric)
\begin{align*}
&\mathbb{E}_{P_{Y^n|X^n}}\left[ 
\ell \left(x^n,Y^n;x_Q\right) | X^n = x^n\right]\\
& = \mathbb{E}_{P_{Y^n|X^n}}\left[ \sum_y \left( \underline{\hat{\alpha}}^\mathrm{T} f(x_Q,y) - \underline{\alpha}^\mathrm{T}f(x_Q,y) + R(x_Q,y)\right)^2 | X^n=x^n\right]\\
& = \mathbb{E}_{P_{Y^n|X^n}}\left[ \left(\underline{\hat{\alpha}} - \underline{\alpha}\right)^\mathrm{T} \mathbf{F}(x_Q) \left(\underline{\hat{\alpha}} - \underline{\alpha}\right) | X^n=x^n \right]\\
& \quad + 2\mathbb{E}_{P_{Y^n|X^n}}\left[ \sum_y  \left(\underline{\hat{\alpha}}-\underline{\alpha} \right)^\mathrm{T} f(x_Q,y) R(x_Q,y) | X^n=x^n\right] + \sum_y R^2(x_Q,y)\\
& = \mathrm{tr}\left\{ \mathbf{F}_n^{\dagger}(x^n)\mathbf{F}(x_Q)\mathbf{F}_n^{\dagger}(x^n) \mathsf{F}(x^n) \right\}  - 2\underline{\alpha}^\mathrm{T}\mathbf{F}(x_Q)\mathbf{F}_n^{\dagger}(x^n)\mathbf{F}_n(x^n)\underline{\alpha} + \underline{\alpha}^{\mathrm{T}} \mathbf{F}(x_Q)\underline{\alpha}\\
& \quad -\frac2n \mathsf{R}^\mathrm{T}(x^n)f^\mathrm{T}(x^n)\mathbf{F}_n^{\dagger}(x^n)\mathbf{F}(x_Q)\underline{\alpha} - 2 \mathsf{R}^\mathrm{T}(x_Q)f^\mathrm{T}(x_Q)\mathbf{F}_n^\perp(x^n)\underline{\alpha}\\  
& \quad + \frac2n \mathsf{R}^\mathrm{T}(x_Q)f^\mathrm{T}(x_Q)\mathbf{F}_n^{\dagger}(x^n)f(x^n)\mathsf{R}(x^n) + \sum_y R^2(x_Q,y).
\end{align*}
\end{proof}

\begin{lemma}\label{lem:aFa_comple_incomplete_insufficient}
For an insufficient demonstrations and an incomplete model, it holds that
\begin{align*}
& \frac1{n^2}\mathrm{tr}\left\{ \mathbf{F}_n^{\dagger}(x^n)\mathbf{F}(x_Q)\mathbf{F}_n^{\dagger}(x^n) f(x^n)\left(\sum_{i,j}v_iv_j^\mathrm{T}\right) f^\mathrm{T}(x^n)\right\}\\
& = \left(\mathbf{F}_n^{\dagger}(x^n)\mathbf{F}_n(x^n)\underline{\alpha}\right)^\mathrm{T}\mathbf{F}(x_Q) \mathbf{F}_n^{\dagger}(x^n)\mathbf{F}_n(x^n)\underline{\alpha}\\
& \quad + \frac2n \mathsf{R}^\mathrm{T}(x^n)f^\mathrm{T}(x^n)\mathbf{F}_n^{\dagger}(x^n) \mathbf{F}(x_Q)\mathbf{F}_n^{\dagger}(x^n) \mathbf{F}_n(x^n)\underline{\alpha}\\
&\quad + \frac1{n^2}  \mathsf{R}^\mathrm{T}(x^n)f^\mathrm{T}(x^n)\mathbf{F}_n^{\dagger}(x^n) \mathbf{F}(x_Q)\mathbf{F}_n^{\dagger}(x^n) f(x^n) \mathsf{R}(x^n).
\end{align*}
\end{lemma}
\begin{proof}
We have that
\begin{align*}
& \frac1{n^2}\mathrm{tr}\left\{ \mathbf{F}_n^{\dagger}(x^n)\mathbf{F}(x_Q)\mathbf{F}_n^{\dagger}(x^n) f(x^n)\left(\sum_{i,j}v_iv_j^\mathrm{T}\right) f^\mathrm{T}(x^n)\right\}\\
& = \mathrm{tr}\left\{ \mathbf{F}_n^{\dagger}(x^n)\mathbf{F}(x_Q)\mathbf{F}_n^{\dagger}(x^n)\left(\frac1n\sum_{i,y}P_{Y|X}(y|x_i) f(x_i, y)\right)\left(\frac1n\sum_{j,y} P_{Y|X}(y|x_j) f(x_j,y)\right)^\mathrm{T}\right\}\\
\end{align*}
From model \ref{equ:model}, we have that for an incomplete model, $P_{Y|X}(y|x) = \underline{\alpha}^\mathrm{T}f(x,y) + R(x,y)$. Hence, we have that
\begin{align*}
& \mathrm{tr}\left\{ \mathbf{F}_n^{\dagger}(x^n)\mathbf{F}(x_Q)\mathbf{F}_n^{\dagger}(x^n)\left(\frac1n\sum_{i,y}P_{Y|X}(y|x_i) f(x_i, y)\right)\left(\frac1n\sum_{j,y} P_{Y|X}(y|x_j) f(x_j,y)\right)^\mathrm{T}\right\}\\
& = \mathrm{tr}\left\{\mathbf{F}_n^{\dagger}(x^n) \mathbf{F}(x_Q)\mathbf{F}_n^{\dagger}(x^n)\left(\frac1n\sum_{i,y} f(x_i, y)(f^\mathrm{T}(x_i, y) \underline{\alpha}+R(x_i,y)) \right)\right.\\
&\quad \quad\left.\left(\frac1n\sum_{j,y} f(x_j,y)(f^\mathrm{T}(x_j,y)\underline{\alpha}+R(x_i,y))\right)^\mathrm{T}\right\}\\
& = \mathrm{tr}\left\{\mathbf{F}_n^{\dagger}(x^n) \mathbf{F}(x_Q)\mathbf{F}_n^{\dagger}(x^n)\mathbf{F}_n(x^n) \underline{\alpha} \left(\mathbf{F}_n(x^n)\underline{\alpha}\right)^\mathrm{T}\right.\\
&\quad + \frac2n \mathbf{F}_n^{\dagger}(x^n) \mathbf{F}(x_Q)\mathbf{F}_n^{\dagger}(x^n) \mathbf{F}_n(x^n)\underline{\alpha} f^\mathrm{T}(x^n)\mathsf{R}^\mathrm{T}(x^n)\\
& \quad \left.+ \frac1{n^2} \mathbf{F}_n^{\dagger}(x^n) \mathbf{F}(x_Q)\mathbf{F}_n^{\dagger}(x^n) f(x^n) \mathsf{R}(x^n)f^\mathrm{T}(x^n)\mathsf{R}^\mathrm{T}(x^n)\right\}\\
& = \left(\mathbf{F}_n^{\dagger}(x^n)\mathbf{F}_n(x^n)\underline{\alpha}\right)^\mathrm{T}\mathbf{F}(x_Q) \mathbf{F}_n^{\dagger}(x^n)\mathbf{F}_n(x^n)\underline{\alpha}\\
& \quad + \frac2n \mathsf{R}^\mathrm{T}(x^n)f^\mathrm{T}(x^n)\mathbf{F}_n^{\dagger}(x^n) \mathbf{F}(x_Q)\mathbf{F}_n^{\dagger}(x^n) \mathbf{F}_n(x^n)\underline{\alpha}\\
&\quad + \frac1{n^2}  \mathsf{R}^\mathrm{T}(x^n)f^\mathrm{T}(x^n)\mathbf{F}_n^{\dagger}(x^n) \mathbf{F}(x_Q)\mathbf{F}_n^{\dagger}(x^n) f(x^n) \mathsf{R}(x^n).\\
\end{align*}
This completes the proof.
\end{proof}

Hence, we have the excessive risk in a simplified form.
\begin{corollary}\label{col:excessive_risk_incomplete_insufficient_simple}
Let 
\[
\textbf{A}(x_i) = \mathsf{diag}\{P_{Y|X}(1|x_i),\dots,P_{Y|X}(M|x_i)\}-\phi_i\phi_i^\mathrm{T},
\]
where $\phi_i = [P_{Y|X}(1|x_i),\dots,P_{Y|X}(M|x_i)]^\mathrm{T}$, and let
\begin{equation*}
\textbf{A}(x^n) = \mathsf{diag}\left\{ \textbf{A}(x_1), \ldots, \textbf{A}(x_n) \right\},
\end{equation*}
it holds that
\begin{align*}
& \mathbb{E}_{P_{Y^n|X^n}}\left[ 
\ell \left(x^n,Y^n;x_Q\right) | X^n = x^n\right]\\
& = \mathrm{tr}\left\{ \frac1{n^2}\mathbf{F}_n^{\dagger}(x^n)\mathbf{F}(x_Q)\mathbf{F}_n^{\dagger}(x^n) f(x^n) {\mathbf{A}}(x^n) f^\mathrm{T}(x^n) \right\}\\
& \quad +
\left\| f ^{\mathrm{T}}(x_Q)\mathbf{F}_n^\perp(x^n)\underline{\alpha} - \mathsf{R}(x_Q) - \frac1n f^\mathrm{T}(x_Q)\mathbf{F}_n^{\dagger}(x^n)f(x^n)\mathsf{R}(x^n)\right\|^2,
\end{align*}
where $\mathbf{F}_n^{\perp}(x^n) = \textbf{I}_K-\mathbf{F}_n^\dagger(x^n) \mathbf{F}_n(x^n)$.
\end{corollary}
\begin{proof}
    Based on the Corollary \ref{col:excessive_risk_complete} and Lemma \ref{lem:aFa_comple}, we have the excessive risk being
\begin{align*}
& \mathbb{E}_{P_{Y^n|X^n}}\left[ 
\ell \left(x^n,Y^n;x_Q\right) | X^n = x^n\right]\\
&= \mathrm{tr}\left\{ \mathbf{F}_n^{\dagger}(x^n)\mathbf{F}(x_Q)\mathbf{F}_n^{\dagger}(x^n) \mathsf{F}(x^n) \right\} - 2\underline{\alpha}^\mathrm{T}\mathbf{F}(x_Q)\mathbf{F}_n^{\dagger}(x^n)\mathbf{F}_n(x^n)\underline{\alpha} + \underline{\alpha}^{\mathrm{T}} \mathbf{F}(x_Q)\underline{\alpha}.\\
& = \mathrm{tr}\left\{ \frac1{n^2}\mathbf{F}_n^{-1}(x^n)\mathbf{F}(x_Q)\mathbf{F}_n^{-1}(x^n) f(x^n) \left(\tilde{\mathbf{A}}(x^n) - \sum_{i,j}v_iv_j^\mathrm{T}\right)f^\mathrm{T}(x^n) \right\}\\
& \quad + \left(\mathbf{F}_n^{\dagger}(x^n)\mathbf{F}_n(x^n)\underline{\alpha}\right)^\mathrm{T}\mathbf{F}(x_Q)\mathbf{F}_n^{\dagger}(x^n)\mathbf{F}_n(x^n)\underline{\alpha} - 2\underline{\alpha}^\mathrm{T}\mathbf{F}(x_Q)\mathbf{F}_n^{\dagger}(x^n)\mathbf{F}_n(x^n)\underline{\alpha} + \underline{\alpha}^{\mathrm{T}} \mathbf{F}(x_Q)\underline{\alpha}\\
& \quad + \frac2n \mathsf{R}^\mathrm{T}(x^n)f^\mathrm{T}(x^n)\mathbf{F}_n^{\dagger}(x^n) \mathbf{F}(x_Q)\mathbf{F}_n^{\dagger}(x^n) \mathbf{F}_n(x^n)\underline{\alpha}\\
& \quad + \frac1{n^2}  \mathsf{R}^\mathrm{T}(x^n)f^\mathrm{T}(x^n)\mathbf{F}_n^{\dagger}(x^n) \mathbf{F}(x_Q)\mathbf{F}_n^{\dagger}(x^n) f(x^n) \mathsf{R}(x^n)\\
& \quad - \frac2n \mathsf{R}^\mathrm{T}(x^n)f^\mathrm{T}(x^n)\mathbf{F}_n^{\dagger}(x^n)\mathbf{F}(x_Q)\underline{\alpha} - 2 \mathsf{R}^\mathrm{T}(x_Q)f^\mathrm{T}(x_Q)\mathbf{F}_n^\perp(x^n)\underline{\alpha}\\  
& \quad + \frac2n \mathsf{R}^\mathrm{T}(x_Q)f^\mathrm{T}(x_Q)\mathbf{F}_n^{\dagger}(x^n)f(x^n)\mathsf{R}(x^n) + \sum_y R^2(x_Q,y).
\end{align*}
Note that $\tilde{\mathbf{A}}(x^n) = \sum_{i,j}v_iv_j^\mathrm{T} + \mathbf{A}(x^n)$ and $\mathbf{F}(x_Q) = f(x_Q)f^\mathrm{T}(x_Q)$, we have that
\begin{align*}
& \mathbb{E}_{P_{Y^n|X^n}}\left[ 
\ell \left(x^n,Y^n;x_Q\right) | X^n = x^n\right]\\
& = \mathrm{tr}\left\{ \frac1{n^2}\mathbf{F}_n^{\dagger}(x^n)\mathbf{F}(x_Q)\mathbf{F}_n^{\dagger}(x^n) f(x^n) {\mathbf{A}}(x^n) f^\mathrm{T}(x^n) \right\}\\
& \quad +
\left\| f ^{\mathrm{T}}(x_Q)\mathbf{F}_n^\perp(x^n)\underline{\alpha} - \mathsf{R}(x_Q) - \frac1n f^\mathrm{T}(x_Q)\mathbf{F}_n^{\dagger}(x^n)f(x^n)\mathsf{R}(x^n)\right\|^2.
\end{align*}
\end{proof}

Following Lemma \ref{lem:better_ub_complete_insufficient}, we have that
\begin{align*}
& \mathbb{E}_{P_{Y^n|X^n}}\left[ 
\ell \left(x^n,Y^n;x_Q\right) | X^n = x^n\right]\\
&\leq \frac1n \lambda_1\left(\mathbf{F}(x_Q)\mathbf{F}_n^{\dagger}(x^n) \right) S_K({\textbf{A}}(x^n))\\
&\quad +
\left\| f ^{\mathrm{T}}(x_Q)\mathbf{F}_n^\perp(x^n)\underline{\alpha} - \mathsf{R}(x_Q) - \frac1n f^\mathrm{T}(x_Q)\mathbf{F}_n^{\dagger}(x^n)f(x^n)\mathsf{R}(x^n)\right\|^2\\
& \leq \frac{K}{n} \lambda_1\left(\mathbf{F}(x_Q)\mathbf{F}_n^{\dagger}(x^n) \right) \lambda_1({\textbf{A}}(x^n))\\
&\quad +
\left\| f ^{\mathrm{T}}(x_Q)\mathbf{F}_n^\perp(x^n)\underline{\alpha} - \mathsf{R}(x_Q) - \frac1n f^\mathrm{T}(x_Q)\mathbf{F}_n^{\dagger}(x^n)f(x^n)\mathsf{R}(x^n)\right\|^2,
\end{align*}
where $S_K({\textbf{A}}(x^n))$ is the sum of top K eigenvalues.

Taking expectation over $P_{X}$, we have the result as
\begin{align*}
&\mathbb{E}_{P_{X}P_{Y^n|X^n}}\left[ 
\ell \left(x^n,Y^n;x_Q\right) | X^n = x^n\right]\\ 
&\leq \frac{K}{n}\lambda_1 \left(\mathbf{F}_Q\mathbf{F}_n^\dagger(x^n) \right) \lambda_1(\textbf{A}(x^n))\notag\\
& \quad + \frac1n \lambda_1 \big(\mathbf{F}_Q \mathbf{F}_n^\dagger(x^n)\big)\cdot \| \mathsf{R}(x^n) \|^2 + \sum_{x_Q,y}{P_{X}(x_Q)R^2(x_Q,y)}\\
& \quad + \underline{\alpha}^{\mathrm{T}}\mathbf{F}_n^{\perp}(x^n)^{\mathrm{T}} \mathbf{F}_Q\mathbf{F}_n^{\perp}(x^n)\underline{\alpha} - \frac2n \mathsf{R}^{\mathrm{T}}(x^n)f^{\mathrm{T}}(x^n) \mathbf{F}_n^\dagger(x^n)\mathbf{F}_Q\mathbf{F}_n^{\perp}(x^n)\underline{\alpha},
\end{align*}
where $\mathbf{F}_Q = \mathbb{E}_{P_{X}}[\mathbb{F}(x_Q)]$.
This is because $\sum_{x_Q} P_X(x_Q) f^\mathrm{T}(x_Q)\mathsf{R}(x_Q) = 0$.

\subsection{Proof of Lemma \ref{lem:lebel_pred_error}}\label{append:lemma2_proof}

We prove the contrapositive. Fix an index \(s\ge j+1\). We show that if the predicted label equals \(s\) (or more weakly, if \(\hat P_s \ge \hat P_1\)), then necessarily
\[
\ell(x^n,y^n;x_Q)\;\ge\; \tfrac12\big(P_1-P_s\big)^2
\;\ge\; \tfrac12\big(P_1-P_{j+1}\big)^2,
\]
which contradicts the hypothesis. Since this holds for every \(s\ge j+1\), no such \(s\) can be the predicted label, and hence the predicted label must lie in \(\{1,\dots,j\}\). Because for any \(r\le j\) we have \(P_r\ge P_j\), the conclusion follows.
Thus fix \(s\ge j+1\) and suppose \(\hat P_s \ge \hat P_1\). Consider the optimization
\[
\min_{\hat P}\ \sum_{t=1}^M(\hat P_t-P_t)^2
\quad\text{s.t.}\quad \hat P_s - \hat P_1 \ge 0,
\]
where \(\hat P\) ranges over \(\mathbb{R}^M\) (the feasible set for probability vectors only increases the minimal cost, so this relaxation provides a valid lower bound). To obtain a lower bound it suffices to restrict attention to coordinates \(1\) and \(s\) and leave all other coordinates equal to their true values \(P_t\). With this restriction the problem reduces to
\[
\min_{u,v\in\mathbb{R}}\ (u-P_1)^2 + (v-P_s)^2
\quad\text{s.t.}\quad v-u \ge 0.
\]
The minimal value of this two-variable problem under the constraint \(v-u\ge0\) is attained when \(v=u\) (pushing toward equality is best), hence we set \(u=v=\eta\) and minimize
\[
(\eta-P_1)^2 + (\eta-P_s)^2.
\]
This quadratic in \(\eta\) is minimized at \(\eta = \tfrac{P_1+P_s}{2}\), giving the minimum value
\begin{equation*}
    \left(\frac{P_1+P_s}{2}-P_1\right)^2 + \left(\frac{P_1+P_s}{2}-P_s\right)^2
= 2\Big(\frac{P_1-P_s}{2}\Big)^2
= \tfrac12\big(P_1-P_s\big)^2.
\end{equation*}
Therefore any \(\hat P\) with \(\hat P_s\ge\hat P_1\) must satisfy
\[
\sum_{t=1}^M(\hat P_t-P_t)^2 \;\ge\; \tfrac12\big(P_1-P_s\big)^2.
\]
Because \(P_s\le P_{j+1}\) for all \(s\ge j+1\), we have \(\tfrac12(P_1-P_s)^2\ge \tfrac12(P_1-P_{j+1})^2\). Hence if
\[
\ell(x^n,y^n;x_Q) < \tfrac12\big(P_1-P_{j+1}\big)^2,
\]
no index \(s\ge j+1\) can satisfy \(\hat P_s\ge\hat P_1\), i.e. the maximizer \(\hat y_{\max}\) must belong to \(\{1,\dots,j\}\).
This implies
\[
P_{Y|X}(\hat y_{\max}\mid x_Q)\;\ge\; P_j.
\]

\subsection{Proof of Theorem \ref{thm:error_prob}}\label{append:theorem4_proof}

To prove the Theorem \ref{thm:error_prob}, we aim to find the minium value of $\mathbb{E}_{P_{Y^n|X^n}}[P_{Y|X}(\hat{y}_{\max} |x_Q)|X^n=x^n]$, which is equivalent to solve the following problem:
\begin{align*}
\min_{\hat{P}_{Y|X}(\cdot|x_Q)} & \sum_{y^n}P_{Y^n|X^n}(y^n|x^n) P_{Y|X}(\hat{y}_{\max}|x_Q)\\
\text{s.t.}& \sum_{y^n} P_{Y^n|X^n}(y^n|x^n) \ell(x^n,y^n;x_Q) \leq \gamma,\\
& \sum_{y^n} P_{Y^n|X^n}(y^n|x^n) = 1,\\
& \sum_y (\hat{P}_{Y|X}(y|x_Q;y^n) - P_{Y|X}(y|x_Q))^2 = \ell(x^n,y^n;x_Q)\\
& \hat{y}_{\max} = \argmax_y \hat{P}_{Y|X}(y|x_Q;y^n)
\end{align*}

Lemma \ref{lem:lebel_pred_error} provides the theoretical guarantee of the CB-ICL label predictor with respect to different threshold values of the mean-squared risk. Notice that $P_{Y|X}(\hat{y}_{\max}|x_Q) \leq \max_y P_{Y|X}(y|x_Q) = P_1$, and the equality is achieved when $\ell(x^n,y^n;x_Q) < \frac12 \left( P_1 - P_{2}\right)^2$. Therefore, the CB-ICL label predictor is reduced to the Maximum a Posteriori (MAP) decision when the mean-squared risk is small. Moreover, the following Theorem establishes the connection between the excessive risk and the label predicting error probability based on Lemma \ref{lem:lebel_pred_error}.

From Lemma \ref{lem:lebel_pred_error}, we know that when $\frac12 (P_1-P_{j})^2 \leq \ell(x^n,y^n;x_Q)< \frac12 (P_1-P_{j+1})^2$, the minimum value $P_{Y|X}(\hat{y}_{\max}|x_Q)$ can take is $P_j$. Hence, the original problem can turn into a combination problem that $\ell$ only takes value in $\{0,\frac12 (P_1-P_2)^2,\dots,\frac12(P_1-P_M)^2\}$, with the following $P_{Y|X}(\hat{y}_{\max}|x_Q)$ being $P_1, P_2, \dots,P_M$. Denote $\ell_j$ as the discrete variable taking value in $\{0, \frac12 (P_1-P_2)^2,\dots,\frac12(P_1-P_M)^2\}$ with corresponding ${P}_j$ being the discrete variable taking value in $\{P_1, P_2,\dots,P_M\}$. The original problem becomes that assign each $y^n$ to one $j\in\{1,\dots,M\}$ such that $\sum_j\sum_{y^n} P_{Y^n|X^n}(y^n|x^n)  \mathsf{1}_{\sigma(y^n)=j}P_j$ achieves minimum (with $\sigma$ denoting the assign function). Denote the weight assigned to $j$th index as $w_j = \sum_{y^n} P_{Y^n|X^n}(y^n|x^n)  \mathsf{1}_{\sigma(y^n)=j}$, and we have the original problem being
\begin{align*}\label{prob:assign}
\min_{w} & \sum_{j=1}^M  w_j {P}_{j}  \tag*{(P1)}\\
\text{s.t.}& \sum_{j=1}^M w_j\ell_{j} \leq \gamma,\\
& w_j = \sum_{y^n} P_{Y^n|X^n}(y^n|x^n)  \mathsf{1}_{\sigma(y^n)=j}.
\end{align*}

To solve this problem, we consider an approximation to this problem as
\begin{align*}\label{prob:continus_assign}
\min_{w} & \sum_{j=1}^M  w_j {P}_{j}  \tag*{(P2)}\\
\text{s.t.}& \sum_{j=1}^M w_j\ell_{j} \leq \gamma,\\
& \sum_{j=1}^M w_j = 1.
\end{align*}

\begin{lemma}\label{lem:dis_2_con}
Denote the solution to problem \ref{prob:assign} as $\{w^{\ast,1}_j\}_{j=1}^M$ and solution to problem \ref{prob:continus_assign} as $\{w^{\ast,2}_j\}_{j=1}^M$. Then, $\sum_{j=1}^M w^{\ast,1}_j P_j \geq \sum_{j=1}^M w^{\ast,2}_j P_j$.
\end{lemma}
\begin{proof}
Let
\[
\mathcal{S}_1=\Big\{w\in\mathbb{R}^M:\exists\ \sigma:[M]^n\to[M]\ \text{such that } w_j=\sum_{y^n}P_{Y^n|X^n}(y^n|x^n)\mathbf{1}_{\{\sigma(y^n)=j\}}\Big\}
\]
be the feasible set of problem \ref{prob:assign}, and
\[
\mathcal{S}_2=\Big\{w\in\mathbb{R}_+^M:\sum_{j=1}^M w_j=1,\ \sum_{j=1}^M w_j\ell_j\le\gamma\Big\}
\]
be the feasible set of problem \ref{prob:continus_assign}. We first show $\mathcal{S}_1\subseteq\mathcal{S}_2$.

Take any $w\in\mathcal{S}_1$. By definition there exists an assignment $\sigma$ such that
\[
w_j=\sum_{y^n}P_{Y^n|X^n}(y^n|x^n)\mathbf{1}_{\{\sigma(y^n)=j\}}\quad (j=1,\dots,M).
\]
Clearly $w_j\ge0$ for all $j$ and
\[
\sum_{j=1}^M w_j=\sum_{j=1}^M\sum_{y^n}P_{Y^n|X^n}(y^n|x^n)\mathbf{1}_{\{\sigma(y^n)=j\}}
=\sum_{y^n}P_{Y^n|X^n}(y^n|x^n)=1.
\]
Moreover
\[
\sum_{j=1}^M w_j\ell_j
=\sum_{j=1}^M\sum_{y^n}P_{Y^n|X^n}(y^n|x^n)\mathbf{1}_{\{\sigma(y^n)=j\}}\ell_j
=\mathbb{E}_{P_{Y^n|X^n}}\big[\ell(x^n,Y^n;x_Q)\big]\le\gamma,
\]
where the last inequality is exactly the feasibility condition in (P1). Hence $w\in\mathcal{S}_2$, proving $\mathcal{S}_1\subseteq\mathcal{S}_2$.

Now let $f(w)=\sum_{j=1}^M w_j P_j$ be the objective. Since $\mathcal{S}_1\subseteq\mathcal{S}_2$, the minimum of $f$ over the smaller set $\mathcal{S}_1$ cannot be smaller than the minimum over the larger set $\mathcal{S}_2$. Formally,
\[
\min_{w\in\mathcal{S}_1} f(w)\ \ge\ \min_{w\in\mathcal{S}_2} f(w).
\]
Noting that the left-hand side equals $\sum_{j} w^{\ast,1}_j P_j$ and the right-hand side equals $\sum_{j} w^{\ast,2}_j P_j$, the claimed inequality follows.
\end{proof}
By Lemma \ref{lem:dis_2_con}, we transfer solving of problem \ref{prob:assign} to problem \ref{prob:continus_assign}. 

\begin{lemma}\label{lem:extreme_two_support}
The solution to problem \ref{prob:continus_assign}, denoted as $\{w^{\ast,2}_j\}_{j=1}^M$, has at most two nonzero components. Furthermore, if there exists two nonzero components, the two components are adjacency to each other.
\end{lemma}
\begin{proof}
Using Lagrange, we have the dual problem as
\begin{align*}
    \min_{w} \sum_{j} w_j{P}_j + \mu\left(\sum_{j} w_j - 1\right) +\lambda\left(\sum_{j} w_j \ell_j - \gamma\right).
\end{align*}
The Karush-Kuhn-Tucker (KKT) condition \citep{bertsekas1997nonlinear} gives that

\begin{align*}
{P}_j+\mu+\lambda\ell_j\left\{
\begin{aligned}
=0&,~w_j>0\\
>0&,~w_j=0
\end{aligned}\right.
\end{align*}
Therefore for all $j$, we have that $\hat{P}_j+\mu+\lambda\ell_j \geq 0$. In other words, the support of the optimal distribution $P$ is contained in the set of indices that minimize the affine functional
\begin{equation*}
    {P}\mapsto{P}+\lambda\ell.
\end{equation*}
Consequently, the optimal solution must assign positive probability only to those outcomes lying on the lower envelope of the family of affine functions parameterized by $\lambda$.

Define the discrete slopes
\[
s_t := \frac{P_t-P_{t+1}}{\ell_t-\ell_{t+1}},\qquad t=1,\dots,M-1.
\]
Compute explicitly using $\ell_t=\tfrac12(P_1-P_t)^2$:
\[
s_t = -\frac{2}{\,2P_1-P_t-P_{t+1}\,}.
\]
Because $P_t$ is nonincreasing and $t\mapsto(2P_1-P_t-P_{t+1})$ is nondecreasing, we have $s_1<s_2<\cdots<s_{M-1}$ (strict inequality unless ties occur in the $P_t$'s; ties can be handled by tie-breaking but do not affect the argument). 

Now suppose there exist positive $w_i>0$ and $w_k>0$ with $k\ge i+2$ (not adjacent). Since $P_t+\lambda\ell_t$ is affine in the pair $(P_t,\ell_t)$ and the intersection equality above holds for $t=i$ and $t=k$, by intermediate value there must exist an index $r$ with $i<r<k$ such that $P_r+\lambda\ell_r$ is \emph{strictly smaller} than the common value (because the sequence of slopes $s_t$ is strictly increasing, the line through $(\ell_i,P_i)$ and $(\ell_k,P_k)$ lies strictly above at some intermediate lattice point). But then $r$ would yield a strictly smaller $P_r+\lambda\ell_r$, contradicting the KKT condition that all positive-weight indices minimize $P_t+\lambda\ell_t$. Therefore no two positive indices can be non-adjacent; positive indices must be adjacent.
\end{proof}

\begin{lemma}\label{lem:alpha}
If mass is placed only at $\ell_j$ and $\ell_{j+1}$ with weights $1-\alpha$ and $\alpha$ and the mean loss equals $\ell_j+\gamma$ (with $0\le\gamma<\ell_{j+1}-\ell_j$), then
\[
\alpha=\frac{\gamma}{\ell_{j+1}-\ell_j}.
\]
\end{lemma}

\begin{proof}
Immediate from $(1-\alpha)\ell_j+\alpha\ell_{j+1}=\ell_j+\gamma$.
\end{proof}

By Lemma \ref{lem:extreme_two_support} any extreme minimizer has at most two adjacent nonzero elements. Therefore the minimizer can be taken with support $\{\ell_j,\ell_{j+1}\}$. Let the mass at $\ell_{j+1}$ be $\alpha$; by Lemma~\ref{lem:alpha} we have $\alpha=\gamma/(\ell_{j+1}-\ell_j)$. Consequently
\[
\mathbb{E}\!\left[P_{Y|X}(\hat y_{\max}\mid x_Q)\right]
\ge (1-\alpha)P_j+\alpha P_{j+1}
= P_j - \alpha(P_j-P_{j+1}).
\]
For each $j$,
\[
\ell_{j+1}-\ell_j
=\tfrac12\big[(P_1-P_{j+1})^2-(P_1-P_j)^2\big]
= (P_j-P_{j+1})\Big(P_1-\tfrac{P_j+P_{j+1}}{2}\Big),
\]
hence
\[
\frac{P_j-P_{j+1}}{\ell_{j+1}-\ell_j}
=\frac{2}{\,2P_1-P_j-P_{j+1}\,}.
\]
Therefore,
\[
\alpha(P_j-P_{j+1})
=\gamma\cdot\frac{P_j-P_{j+1}}{\ell_{j+1}-\ell_j}
=\gamma\cdot\frac{2}{\,2P_1-P_j-P_{j+1}\,}.
\]
Thus
\[
\mathbb{E}\!\left[P_{Y|X}(\hat y_{\max}\mid x_Q)\right]
\ge P_j - \frac{2\gamma}{\,2P_1-P_j-P_{j+1}\,},
\]
completing the proof.

\section{Proofs of Properties}

\subsection{Bound of Similarity Score}\label{sec:lb_lambda_1}
\begin{theorem}\label{thm:lambda_ge_1}
Let
\[
f(x,y)\in\mathbb R^K,\qquad f(x)=[\,f(x,1),\dots,f(x,M)\,]\in\mathbb R^{K\times M},
\]
and define
\[
\mathbf F(x_Q)=f(x_Q)f(x_Q)^\mathrm{T},\qquad
\mathbf F_n(x^n)=\frac{1}{n}f(x^n)f(x^n)^\mathrm{T}
=\frac{1}{n}\sum_{i=1}^n\sum_{y=1}^M f(x_i,y)f(x_i,y)^\mathrm{T}.
\]
Assume \(\mathbf F_n(x^n)\) is positive definite. If the last coordinate satisfies
\(f_K(x,y)=1/\sqrt{M}\) for every \(x,y\), then
\[
\lambda_1\!\big(\mathbf F(x_Q)\mathbf F_n^{-1}(x^n)\big)\;\ge\;1,
\]
where \(\lambda_1(\cdot)\) denotes the largest eigenvalue.
\end{theorem}
\begin{proof}
Recall the characterization of the largest generalized eigenvalue for symmetric matrices:
\[
\lambda_1\!\big(\mathbf F(x_Q)\,\mathbf F_n^{-1}(x^n)\big)
\;=\;
\sup_{w\in\mathbb R^K\setminus\{0\}}
\frac{w^\mathrm{T} \mathbf F(x_Q) w}{w^\mathrm{T} \mathbf F_n(x^n) w}.
\]
Let \(e_K\in\mathbb R^K\) be the unit vector with 1 in the \(K\)-th coordinate and zeros elsewhere.
Using \(f_K(x,y)=1/\sqrt{M}\) for every \(x,y\), we compute
\[
e_K^\mathrm{T} \mathbf F(x_Q) e_K
= \sum_{y=1}^M \big(f_K(x_Q,y)\big)^2
= \sum_{y=1}^M \frac{1}{M} = 1,
\]
and
\[
e_K^\mathrm{T} \mathbf F_n(x^n) e_K
= \frac{1}{n}\sum_{i=1}^n\sum_{y=1}^M \big(f_K(x_i,y)\big)^2
= \frac{1}{n}\sum_{i=1}^n\sum_{y=1}^M \frac{1}{M}
= 1.
\]
Hence the Rayleigh quotient at \(e_K\) equals \(1\):
\[
\frac{e_K^\mathrm{T} \mathbf F(x_Q) e_K}{e_K^\mathrm{T} \mathbf F_n(x^n) e_K}=1.
\]
Since the supremum over all nonzero \(w\) is at least the value at \(e_K\), we obtain
\[
\lambda_1\!\big(\mathbf F(x_Q)\mathbf F_n^{-1}(x^n)\big)\ge 1,
\]
as required.
\end{proof}

\begin{proposition}[When equality holds]\label{prop:when_eq}
With the notation and assumptions of Theorem \ref{thm:lambda_ge_1}, we have
\[
\lambda_1\!\big(\mathbf F(x_Q)\mathbf F_n^{-1}(x^n)\big)=1
\]
if and only if
\[
\mathbf F(x_Q)\preceq \mathbf F_n(x^n),
\]
i.e. \(\mathbf F_n(x^n)-\mathbf F(x_Q)\) is positive semidefinite. In particular, a simple sufficient condition for equality is
\[
\mathbf F_n(x^n)=\mathbf F(x_Q),
\]
which occurs for example when \(n=1\) and \(x_1=x_Q\), or more generally when every sample equals \(x_Q\) (i.e. \(x_1=\cdots=x_n=x_Q\)).
\end{proposition}

\begin{proof}
By the Rayleigh characterization,
\[
\lambda_1\!\big(\mathbf F(x_Q)\mathbf F_n^{-1}(x^n)\big)
=\sup_{w\ne0}\frac{w^\mathrm{T} \mathbf F(x_Q) w}{w^\mathrm{T} \mathbf F_n(x^n) w}.
\]
Equality \(\lambda_1=1\) holds iff for every nonzero \(w\),
\[
\frac{w^\mathrm{T} \mathbf F(x_Q) w}{w^\mathrm{T} \mathbf F_n(x^n) w}\le 1,
\quad\text{i.e.}\quad
w^\mathrm{T} \mathbf F(x_Q) w \le w^\mathrm{T} \mathbf F_n(x^n) w.
\]
The last inequality for all \(w\) is exactly the PSD ordering \(\mathbf F(x_Q)\preceq\mathbf F_n(x^n)\). Hence equality is equivalent to \(\mathbf F(x_Q)\preceq\mathbf F_n(x^n)\).

If \(\mathbf F_n(x^n)=\mathbf F(x_Q)\) then trivially \(\mathbf F(x_Q)\preceq\mathbf F_n(x^n)\) and thus \(\lambda_1=1\). The condition \(\mathbf F_n(x^n)=\mathbf F(x_Q)\) holds when \(f(x_i,y)=f(x_Q,y)\) for every \(i,y\), i.e. when \(x_i=x_Q\) for all \(i\); in particular it holds when \(n=1\) and \(x_1=x_Q\). This proves the stated sufficient condition.
\end{proof}

\end{document}